%% file: main.tex
\documentclass[letterpaper,11pt]{article}
\usepackage[utf8]{inputenc}

\usepackage{amsmath, amsthm, amssymb}
\usepackage{algorithmicx}
\usepackage[table,xcdraw]{xcolor}
\usepackage[ruled,vlined,linesnumbered]{algorithm2e}
\usepackage{mathtools}
\usepackage[numbers]{natbib}
\usepackage{comment} 
\usepackage{tcolorbox} 
\usepackage{xfrac}
\usepackage{hyperref}
\usepackage{multirow}
\usepackage{caption}
\usepackage{bm}
\usepackage{newfloat}
\usepackage{enumitem}

\usepackage{fancyhdr}

\usepackage[margin=1in]{geometry}

\allowdisplaybreaks

\definecolor{mygreen}{RGB}{20,150,80}
\definecolor{myred}{RGB}{150,0,0}

\hypersetup{
     colorlinks=true,
     citecolor= mygreen,
     linkcolor= myred
}

\input{macro.tex}
\input{utils.tex}

\makeatletter
\renewcommand{\paragraph}{%
  \@startsection{paragraph}{4}%
  {\z@}{10pt}{-1em}%
  {\normalfont\normalsize\bfseries}%
}
\makeatother

\makeatletter
\patchcmd{\@algocf@start}
  {-1.5em}
  {0pt}
  {}{}
\makeatother

\title{New Trade-Offs for Fully Dynamic Matching\\ via Hierarchical EDCS}

\author{
Soheil Behnezhad\\{\em Stanford University}\\ \texttt{beh@cs.stanford.edu} \\ \and 
Sanjeev Khanna\thanks{Supported in part by NSF awards CCF-1763514, CCF-1934876, and CCF-2008305.}
\\{\em University of Pennsylvania} \\  \texttt{sanjeev@cis.upenn.edu} \\
}

\date{}

\begin{document}

\maketitle

\input{abstract}

\clearpage

\setcounter{tocdepth}{2}
\tableofcontents

\clearpage

\input{intro.tex}

\input{randDynamicEDCS}

\input{conclusion}

\clearpage

\bibliographystyle{alpha}
\bibliography{references}

\appendix
\clearpage

\input{analysis-of-fkb}

\input{appendix}

\end{document}

%% file: macro.tex
\renewcommand{\epsilon}{\varepsilon}
\newcommand{\stackeq}[2]{\ensuremath{\stackrel{\text{#1}}{#2}}}

\DeclareMathOperator{\E}{\ensuremath{\normalfont \textbf{E}}}

\newcommand{\addLayerNoInput}[0]{\ensuremath{\textsc{AddLayer}}}
\newcommand{\modifiedAddLayerNoInput}[0]{\ensuremath{\textsc{ModifiedAddLayer}}}
\newcommand{\computeLayersNoInput}[0]{\ensuremath{\textsc{ComputeLayers}}}

\newcommand{\PreProcessing}[1]{\ensuremath{\textsc{PreProcessing}(#1)}}
\newcommand{\addLayer}[1]{\ensuremath{\textsc{AddLayer}(#1)}}

\newcommand{\computeLayers}[1]{\ensuremath{\textsc{ComputeLayers}(#1)}}

\newcommand{\unifzeroone}[0]{\ensuremath{\mathsf{Unif}[0, 1]}}

\newcommand{\hiddencomment}[1]{}

\newcommand{\HEDCS}[2]{\ensuremath{#1\text{-}\text{\normalfont HEDCS}_{#2}}}

\newcommand{\mc}[1]{\ensuremath{\mathcal{#1}}}

\newcommand{\maximalmatchingsize}[1]{\ensuremath{\widetilde{\mu}(#1)}}

\newcommand{\APXKTWOBIPARTITE}[0]{.612}
\newcommand{\APXKTWOGENERAL}[0]{.609}
\newcommand{\APXKTHREEBIPARTITE}[0]{.563}
\newcommand{\APXKTHREEGENERAL}[0]{.532}

%% file: utils.tex
\DeclareMathOperator{\poly}{poly}

\newcommand{\Ot}{\ensuremath{\widetilde{O}}}

\usepackage[noabbrev,nameinlink]{cleveref}

\newtheorem{theorem}{Theorem}[section]
\newtheorem{lemma}[theorem]{Lemma}
\newtheorem{proposition}[theorem]{Proposition}
\newtheorem{corollary}[theorem]{Corollary}

\newtheorem{definition}[theorem]{Definition}
\newtheorem{claim}[theorem]{Claim}

\newtheorem{observation}[theorem]{Observation}
\newtheorem{remark}[theorem]{Remark}

\newtheorem{open}{Open Problem}

\definecolor{mylightgray}{RGB}{230,230,230}

\algnewcommand{\IIf}[1]{\State\algorithmicif\ #1\ \algorithmicthen}
\algnewcommand{\EndIIf}{\unskip\ \algorithmicend\ \algorithmicif}

\newenvironment{highlighttechnical}{
\par\addvspace{0.1cm}
\begin{tcolorbox}[width=\textwidth,
                  boxsep=5pt,
                  left=1pt,
                  right=1pt,
                  top=2pt,
                  bottom=2pt,
                  boxrule=1pt,
                  arc=0pt,
                  colframe=black
                  ]
}{
\end{tcolorbox}
}

\newenvironment{myproof}{
\vspace{-0.5cm}
\begin{proof}
}{
\end{proof}
}

\newenvironment{restate}[1]{\textbf{#1} (restated)\textbf{.}\begin{itshape}}{\end{itshape}\vspace{0.3cm}}

%% file: abstract.tex
\begin{abstract}
	We study the maximum matching problem in {\em fully dynamic} graphs: a graph is undergoing both edge insertions and deletions, and the goal is to efficiently maintain a large matching after each edge update. This problem has received considerable attention in recent years. The known algorithms naturally exhibit a trade-off between the quality of the matching maintained (i.e., the approximation ratio) and the time needed per update. While several interesting results have been obtained, the optimal behavior of this trade-off remains largely unclear. Our main contribution is a new approach to designing fully dynamic approximate matching algorithms that in a {\em unified manner} not only (essentially) recovers all previously known trade-offs that were achieved via very different techniques, but reveals some new ones as well.
	
	\medskip
	Specifically, we introduce a generalization of the {\em edge-degree constrained subgraph} (EDCS) of Bernstein and Stein (2015) that we call the {\em hierarchical EDCS} (HEDCS). We also present a randomized algorithm for efficiently maintaining an HEDCS. In an $m$-edge graph with maximum degree $\Delta$, for any integer $k \geq 0$ that is essentially the number of levels of the hierarchy in HEDCS, our algorithm takes $\Ot(\min\{\Delta^{1/(k+1)}, m^{1/(2k+2)}\})$ worst-case update-time and maintains an (almost) $\alpha(k)$-approximate matching where we show:
	\begin{itemize}[itemsep=5pt,topsep=5pt]
		\item $\alpha(0) = 1$, $\alpha(1) = \frac{2}{3}$, $\alpha(\tfrac{1}{\delta}) \geq (\tfrac{1}{2} + \Omega_\delta(1))$ for any $\delta > 0$, and $\alpha(\log \Delta) \geq \frac{1}{2}$.\\[0.1cm]
		These bounds recover all previous trade-offs known for dynamic matching in the literature up to logarithmic factors in the update-time.
		\item $\alpha(2) > \APXKTWOBIPARTITE$ for bipartite graphs, and $\alpha(2) > \APXKTWOGENERAL$ for general graphs.\\[0.1cm]
		Note that these approximations are obtained in $\widetilde{O}(\min\{\Delta^{1/3}, m^{1/6}\})$ update-time.
		\item $\alpha(3) > \APXKTHREEBIPARTITE$ for bipartite graphs, and $\alpha(3) > \APXKTHREEGENERAL$ for general graphs.\\[0.1cm]
		Note that these approximations are obtained in $\widetilde{O}(\min\{\Delta^{1/4}, m^{1/8}\})$ update-time.
	\end{itemize}
\end{abstract}

%% file: intro.tex
\clearpage
\section{Introduction}

The maximum matching problem in graphs plays a central role in combinatorial optimization, and hence has been extensively studied in the classical setting where we are given a static graph, and the goal is to compute a maximum matching of the graph. However, in many applications of the matching problem, the input graph may be dynamically changing via edge insertions and deletions. A natural question is if it is possible to efficiently maintain a near-optimal matching of a graph as it undergoes changes. In this paper, we study design of efficient {\em fully dynamic} algorithms for the {\em maximum matching} problem. Specifically, given a graph $G$ that undergoes both edge insertions and deletions, the goal is to maintain a matching of $G$ that
\begin{enumerate}[label=$(\roman*)$,topsep=0pt,itemsep=0pt]
	\item is approximately as large as the maximum matching of $G$ after each update, and
	\item the time-complexity needed to address each update is small.
\end{enumerate}

Throughout the paper, we will denote by $n$ the number of vertices in $G$, by $m$ the maximum number of edges in $G$ at any time, and by $\Delta$ the maximum degree in $G$ at any time. We will say that a matching is an {\em $\alpha$-approximate} for some $\alpha \in [0,1]$ if its size is at least an $\alpha$-fraction of the maximum matching size.

When the goal is to maintain an {\em exact} maximum matching, then there are conditional lower bounds \cite{AbboudW14,HenzingerKNS15,Dahlgaard-ICALP16} ruling out any $O(n^{1-\epsilon})$ update-time algorithm (see also \cite{Sankowski07,BrandNS19} for some progress on the algorithmic side). As a result, much of the focus in the literature has been on approximate solutions (see e.g.   \cite{BaswanaGS11,BaswanaGS18,NeimanS-STOC13,GuptaPeng-FOCS13,BhattacharyaHI-SiamJC18,BhattacharyaHN-SODA17,BhattacharyaHN-STOC16,Solomon-FOCS16,CharikarS18-ICALP,ArarCCSW-ICALP18,BernsteinFH-SODA19,BehnezhadDHSS-FOCS19,BehnezhadLM-SODA20,Wajc-STOC20,BernsteinDL-STOC21,BhattacharyaK21-ICALP21,RoghaniSW21-arXiv} and the references therein). These works indeed show that settling for an approximate matching does translate into much improved update times.
We highlight some of these results below.

Baswana, Gupta, and Sen~\cite{BaswanaGS11, BaswanaGS18} showed that a maximal matching, and hence a $1/2$-approximate matching, can be maintained in $O(\log n)$ amortized update time (see also the follow-up work by Solomon \cite{Solomon-FOCS16}). At a high-level, their algorithm is based on the insight that if a vertex of degree $d$ is matched to a random neighbor, then in expectation, $\Omega(d)$ updates need to occur before this edge is affected.
Gupta and Peng~\cite{GuptaPeng-FOCS13} showed that for any fixed $\epsilon > 0$, a $(1-\epsilon)$-approximate matching can be maintained with a worst-case update time of $O(\min\{\Delta, m^{1/2}\})$. At a high-level, their approach is based on recomputing a $(1-\epsilon)$-approximate matching once a sufficiently large number of updates have occurred. This idea directly gives the desired result when the matching size is large, and the authors then show that whenever the matching size is small, the underlying graph can be appropriately reduced in size. Bernstein and Stein~\cite{BernsteinSteinICALP15,BernsteinSteinSODA16} introduced a powerful data structure called the {\em edge-degree constrained subgraph} (EDCS), a sparse subgraph of the original graph guaranteed to contain an (almost) $2/3$-approximate matching. The authors showed that this data structure can be maintained with $O(\min\{\Delta^{1/2}, m^{1/4}\})$ amortized update time, yielding a much faster algorithm for maintaining a $2/3$-approximate matching. 
A different approach, based on augmenting a half-approximation using short augmenting paths, was subsequently used by~\cite{BehnezhadLM-SODA20} (see also \cite{BhattacharyaHN-STOC16,Wajc-STOC20}) to show that a $(1/2+\Omega_\epsilon(1))$-approximate matching can be maintained in $\Ot(\Delta^{\epsilon})$ update time for any $\epsilon > 0$.

The results above clearly highlight a trade-off between the approximation ratio of the maintained matching and the update-time. While these results capture many interesting trade-offs, two salient features of the current state of the art are $(i)$ there are many interesting regions where the trade-off between approximation ratio and update time is not understood, and $(ii)$ very different techniques are used in obtaining results at different parts of this trade-off spectrum. 

The main contribution of this work is a new approach to designing fully dynamic algorithms for approximate maximum matching that in a {\em unified manner} not only (essentially) recovers all previously known trade-offs (up to logarithmic factors in the update-time) but reveals some new ones as well. Specifically, we prove the following theorem:

\smallskip

\begin{highlighttechnical}
\begin{theorem}\label{thm:main}
For any integer $k \ge 0$ and any $\epsilon > 0$, there is a randomized algorithm that maintains an $(\alpha(k)-\epsilon)$-approximate maximum matching of a fully dynamic graph against an oblivious adversary with worst-case update time $\min\{\Delta^{1/(k+1)}, m^{1/(2k+2)}\} \cdot \poly(k, 1/\epsilon, \log n)$ where 
\vspace{-0.2cm}
$$
\alpha(0) = 1, \,\, \alpha(1) = 2/3, \,\, \alpha(2) \ge \APXKTWOGENERAL, \,\, \alpha(3) \ge \APXKTHREEGENERAL, \,\, ..., \,\, \alpha(\log \Delta) \geq 1/2.
\vspace{-0.25cm}
$$
If the graph is bipartite, then we show $\alpha(2) \geq \APXKTWOBIPARTITE$ and $\alpha(3) \geq \APXKTHREEBIPARTITE$.
\end{theorem}
\end{highlighttechnical}

Therefore, our algorithm takes as input an integer parameter $k \geq 0$ such that as $k$ goes from $0$ to $\log \Delta$, the update time improves from $\Ot(\min\{\Delta, m^{1/2}\})$ to $\widetilde{O}(1)$. The approximation ratio, on the other hand, goes from almost $1$ to almost $1/2$ as $k$ goes from $0$ to $\log \Delta$.  

Additionally, we prove $\alpha(1/\delta) \geq \frac{1}{2} + \frac{1}{2^{2^{O(1/\delta)}}}$ for any $\delta > 0$. Thus our algorithm can also beat half-approximation with any arbitrarily small polynomial update-time.

\Cref{table:results} summarizes these trade-offs and compares them with prior works.

{
\renewcommand{\arraystretch}{1.4}
\newcommand{\tablecaption}[0]{Approximation/update-time trade-offs of our algorithm for different values of parameter $k$. The algorithm is randomized and the bound on the update-time is worst-case.
By an approximation factor of $\sim \alpha$ we mean the algorithm can get $(1-\epsilon)\alpha$-approximation for any fixed $\epsilon > 0$.}

\begin{table}
\centering
\resizebox{\linewidth}{!}{%
\begin{tabular}{|l|l|l|l|}
\hline
\rowcolor[HTML]{EFEFEF} 
Approximation            & Update-Time                                      & $k$                    & Note                                                                                                      \\ \hline
$\sim 1$                 & $\Ot(\min\{\Delta, m^{1/2}\})$                   & $k=0$                  & This matches \cite{GuptaPeng-FOCS13}.                                                    \\ \hline
$\sim 2/3$               & $\Ot(\min\{\Delta^{1/2}, m^{1/4}\})$             & $k=1$                  & This matches \cite{BernsteinSteinICALP15,BernsteinSteinSODA16}.                         \\ \hline
$\APXKTWOGENERAL$ ($\APXKTWOBIPARTITE$ bipartite)                & $\Ot(\min\{\Delta^{1/3}, m^{1/6}\})$             & $k=2$                  & \textbf{This is a new trade-off.}                                                                         \\ \hline
$\APXKTHREEGENERAL$ ($\APXKTHREEBIPARTITE$ bipartite)                   & $\Ot(\min\{\Delta^{1/4}, m^{1/8}\})$             & $k=3$                  & \textbf{This is a new trade-off.}                                                                         \\ \hline
\multicolumn{4}{|c|}{$\vdots$}                                                                                                                                                                                   \\ \hline
$1/2+\Omega_\delta(1)$ & $\Ot_{\delta}(\min\{\Delta^{\delta}, m^{\delta/2}\})$ & $k=\sfrac{1}{\delta}-1$ & This matches \cite{BehnezhadLM-SODA20}.               \\ \hline
$\sim 1/2$               & $\Ot(1)$                                         & $k = \Theta(\log \Delta)$   & This matches \cite{BaswanaGS11, BaswanaGS18,Solomon-FOCS16}. \\ \hline
\end{tabular}
}
\caption{\tablecaption}
\label{table:results}
\end{table}
}

\subsection{Overview of Techniques}\label{sec:techniques}
\vspace{-0.2cm}

As our main tool, we introduce a generalization of the {\em edge-degree constrained subgraph} (EDCS) of Bernstein and Stein \cite{BernsteinStein15ArXiv, BernsteinSteinSODA16} that we call the {\em hierarchical EDCS} (HEDCS). Before formalizing our generalization HEDCS of EDCS, let us recall the notion of EDCS.

For any edge $e=(u, v)$ in a graph $G$ we use the notation $\deg_G(e) := \deg_G(u) + \deg_G(v)$ to denote the degree of $e$ in $H$, which is simply sum of the degrees of its endpoints. For an integer $\beta > 1$, a subgraph $H$  of a graph $G$ is called a {\em $\beta$-EDCS} of $G$ if:
\begin{enumerate}[label=$(\roman*)$]
	\item For any $e \in H$, $\deg_H(e) \leq \beta$.
	\item For any $e \in G \setminus H$, $\deg_H(e) \geq \beta - 1$.
\end{enumerate}

Interestingly, these two local constraints suffice to guarantee that subgraph $H$ of $G$ includes a $(2/3-O(\epsilon))$-approximate maximum matching of $G$ if $\beta \geq 1/\epsilon$ (see \cite{BehnezhadEDCS21,AssadiBernsteinSOSA19}).

The HEDCS is a hierarchical generalization of EDCS, where in addition to $\beta$, we have a parameter $k$ which is the number of levels in the hierarchy. Formally, the HEDCS is defined as:

\begin{highlighttechnical}
	\begin{definition}[\textbf{Hierarchical Edge-Degree Constrained Subgraphs (HEDCS)}]\label{def:HEDCS}~ 
		
	Let $\beta \geq 2$ and $k \geq 1$ be integers. We say that a graph $H$ is a \HEDCS{\beta}{k} of $G$ iff there is a hierarchical decomposition $\emptyset = H_0 \subseteq H_1 \subseteq H_2 \subseteq \ldots \subseteq H_k = H$ satisfying the following:
\begin{enumerate}[label=$(\roman*)$]
	\item For every $1 \leq i \leq k$ and any edge $e \in H_i \setminus H_{i-1}$, $\deg_{H_i}(e) \leq \beta$.
	\item For any edge $e \in G \setminus H$, $\deg_{H}(e) \geq \beta - 1$.
\end{enumerate}
\end{definition}
\end{highlighttechnical}

Note that a \HEDCS{\beta}{k} for $k=1$ is equivalent to a $\beta$-EDCS. However, as $k$ becomes larger than $1$, the edge-degree constraints in a $\HEDCS{\beta}{k}$ $H$ of $G$ become more relaxed, and for some edges $e$, we may now have $\deg_H(e) > \beta$ (this can, e.g., happen for any edge in $H_{k-1}$). This relaxation of edge-degree constraints gradually weakens the $2/3$-approximation guarantee of an EDCS as $k$ becomes larger but in return, we show that it becomes easier to maintain a $\HEDCS{\beta}{k}$ in dynamic graphs as we increase the number of levels $k$.

\vspace{-0.2cm}
\paragraph{Maintaining an HEDCS:} We give an algorithm that for any $\beta$ and $k$ can (lazily) maintain a \HEDCS{\beta}{k} in a fully dynamic graph with update-time $\widetilde{O}(\min\{\Delta^{1/(k+1)}, m^{1/(2k+2)}\})$. Here, and for the sake of intuition, we will only overview the key insights behind maintaining a \HEDCS{\beta}{k} in  amortized update-time $\widetilde{O}(\Delta^{1/(k+1)})$.

Our definition of \HEDCS{\beta}{k} allows for a $k$-step greedy way of constructing it: construct $H_1$, then construct $H_2$ by adding some edges to $H_1$, then construct $H_3$ by adding edges to $H_2$, and so on so forth. The crucial observation here is that each edge $e \in H_i$ is constrained by property $(i)$ of HEDCS to have edge-degree $\leq \beta$ only in subgraph $H_i$, regardless of which edges belong to $H_{i+1}, \ldots, H_k$. Hence, it is safe to increase the edge-degree of $e \in H_i$ in the higher levels to beyond $\beta$. As a result of this greedy construction, right after constructing $H_i$, any edge $e$ with $\deg_{H_i}(e) \geq \beta - 1$ will for sure satisfy the constraint $\deg_H(e) \geq \beta-1$ of property $(ii)$ of HEDCS. Thus, this edge $e$ can be safely ignored in constructing the higher levels.

To make use of the greedy construction above, we first random sample the edges of $G$ into subgraphs $G_1 \subseteq G_2 \subseteq \ldots \subseteq G_{k+1} = G$, where each $G_i$ includes each edge of $G$ with probability $p_i \approx \Delta^{\frac{i}{k+1} - 1}$. We construct $H_1$ only using the edges of $G_1$, then construct $H_2$ by adding some of the edges of $G_2 \setminus G_1$ to $H_1$, then construct $H_3$ by adding some of the edges of $G_3 \setminus G_2$ to $H_2$, etc. However, instead of considering all the edges in $G_i \setminus G_{i-1}$ in constructing $H_i$, we ignore those edges in $G_i \setminus G_{i-1}$ that are already covered by $H_{i-1}$. To make sure that this helps with pruning the set of edges that we consider in each level, we specifically construct each $H_i$ in a way that guarantees a {\em sparsification property}. That is, the set of edges left uncovered by $H_i$ in the remaining graph, will be in the order $\widetilde{O}(\mu_i / p_i)$ where $\mu_i$ is the size of the largest matching in $G_i$.

One main challenge in maintaining this HEDCS $H$ in a fully dynamic graph is that the edges that are removed from each $H_i$ may result in uncovered edges in the remaining graph, invaliding property $(ii)$ of HEDCS. The crucial observation is that each edge removal of the adversary belongs to $H_i$ (and thus $G_i$) with probability at most $p_i$. Hence, the adversary needs to issue $\approx \epsilon \mu_i/p_i$ updates to remove $\epsilon \mu_i$ edges of $H_i$. We can thus take a lazy approach in maintaining our solution. For every $i \in [k]$, we can roughly speaking ``pretend'' for $\epsilon \mu_i /p_i$ updates that no edge of $H_i$ is removed (i.e., we assume those removed still exist in the graph) and thus all edges covered by $H_i$ remain covered. After $\epsilon \mu_i/p_i$ updates, we recompute all of $H_i, \ldots, H_k$ from scratch, and amortize the cost over these updates. Since $H_{i-1}$, as discussed, only leaves $\widetilde{O}(\mu_{i-1} / p_{i-1})$ edges uncovered, we are able to construct all of $H_i, \ldots, H_k$ in time $\widetilde{O}(\mu_{i-1}/p_{i-1})$. Amortizing this cost over $\epsilon \mu_i/p_i$ updates leads to a bound of $\approx \frac{\mu_{i-1}/p_{i-1}}{\epsilon \mu_i/p_i} \lesssim \frac{p_i}{p_{i-1}} = \Delta^{1/(k+1)}$ update-time for each level $i \in [k]$. 

It is worth noting that our algorithm for maintaining a \HEDCS{\beta}{k} is very different from the algorithms of Bernstein and Stein \cite{BernsteinSteinSODA16} for maintaining an EDCS. In particular, \cite{BernsteinSteinSODA16} maintain a $\beta$-EDCS for $\beta \approx \sqrt{\Delta}$ deterministically and their update-time bound is amortized. In our construction, however, $\beta$ is much smaller and a constant value often suffices (for $k=1$, particularly, where we recover the (almost) $2/3$-approximation of \cite{BernsteinSteinSODA16}, $\beta$ is a constant). Additionally, we use randomization in a crucial way but achieve a worst-case update-time in return.

\vspace{-0.2cm}
\paragraph{Approximation ratio of HEDCS:} To understand the approximation ratio of $\HEDCS{\beta}{k}$, we study a function $\alpha(k)$ which essentially tracks how the native $2/3$-approximation guarantee of EDCS gradually weakens with increasing $k$, when $\beta$ is sufficiently large. The precise analysis of the function $\alpha(k)$ becomes challenging even for small values of $k$. However, for any $\beta$, the computation of $\alpha(k)$ can be expressed as a linear program (LP) (formalized in \Cref{sec:fk-LP}). As a result, for $k=2, 3$, we compute the value of $\alpha(k)$ by solving this LP for a sufficiently large value of $\beta$ which we then feed into our dynamic algorithm.

For larger values of $k$ and $\beta$ sufficiently large with respect to $k$, we analytically compute the value $\alpha(k)$ and show that it is at least $1/2 + \Omega(1/2^{2^{2k}})$ for any $k \ge 1$. This, in particular, means that for any fixed $\epsilon > 0$, the update time can be driven down to $\Ot_\epsilon(\min\{\Delta^{\epsilon}, m^{\epsilon/2}\})$ while still obtaining an approximation ratio that is strictly better than $1/2$, namely, $1/2+\Omega_\epsilon(1)$.

\subsection{Organization}

We start by presenting some notation and relevant results from previous works in \Cref{sec:prelims}.  We then present the hierarchical EDCS (HEDCS) data structure and its properties in \Cref{sec:HEDCS}. We also present here an LP-based approach for analyzing the approximation ratio achieved by the HEDCS data structure, and show the performance implied by it when the number of hierarchy levels $k$ is small. We defer the analysis of approximation achieved by HEDCS data structure when $k$ is allowed to asymptotically grow to \Cref{sec:analysis-of-fkb}.
Finally, in \Cref{sec:MainAlg} we present our randomized algorithm for maintaining the HEDCS data structure, and analyze its performance.

%% file: randDynamicEDCS.tex
\section{Preliminaries}
\label{sec:prelims}

\paragraph{Notation:}
We denote the input graph by $G=(V, E)$. The vertex-set $V$ includes $n$ vertices that are fixed. However, the edge-set $E$ is  dynamic. Particularly, edges can be both inserted and deleted from $E$. We use $\Delta$ as a fixed upper bound on the graph's maximum degree at all times.

All graphs that we define in this work are on the same vertex-set $V$ as the input graph $G$. As such, when it is clear from the context, we may treat these (sub)graphs as essentially sets of edges. Particularly, for a subgraph $H$ we may use $|H|$ to denote the number of edges in $H$, or may use $H \setminus H'$ for two graphs $H$ and $H'$ on vertex set $V$ to denote a subgraph on the same vertex-set, including edges of $H$ that do not belong to $H'$.

For any graph $H$, we use $\mu(H)$ to denote the size of a maximum matching in $H$ and use $\maximalmatchingsize{H}$ to denote the size of a maximal matching of $H$. (We particularly use $\maximalmatchingsize{H}$ when we want the value to be computable in linear time.) For any edge $e = (u, v)$ we define $\deg_H(e) := \deg_H(u) + \deg_H(v)$ to be the {\em edge-degree} of $e$ in graph $H$. We note that so long as the endpoints of $e$ belong to the vertex-set of $H$, $\deg_H(e)$ is well-defined and $e$ does not need to belong to the edge-set of $H$. For any edge $e=(u, v)$, we say $e$ is {\em $(H, \beta)$-underfull} if $\deg_H(e) < \beta - 1$ and {\em $(H, \beta)$-overfull} if $\deg_H(e) > \beta$.

Throughout the paper, the $\Ot(f)$ notation hides $\poly(\log n)$ factors, i.e., $\Ot(f) = f \cdot \poly(\log n)$.

\paragraph{Oblivious Adversary and Worst-Case Update-Time:}
Our dynamic algorithms are based on the standard {\em oblivious adversary assumption}. The sequence of updates in this model are provided by a computationally unbounded adversary that knows the algorithm to be used. However, the updates should not depend on the coin flips of the dynamic algorithm. Equivalently, one may assume that the sequence of updates are  fixed before the dynamic algorithm starts to operate.

As standard, we say a fully dynamic algorithm has ``worst-case update-time'' $T$ if every update is w.h.p. addressed in $T$ time by the algorithm.

\paragraph{Tools:} We will use the following algorithms from prior work.

\begin{proposition}[\cite{MicaliV-FOCS80,Vazirani-Arxiv12}]\label{prop:apxmatching}
	Given any $m$-edge graph $G=(V, E)$ and any parameter $\epsilon > 0$, there is a static algorithm to find a $(1-\epsilon)$-approximate maximum matching of $G$ in $O(m/\epsilon)$ time.
\end{proposition}

In our algorithm, we will need a subroutine that maintains a $c$-approximation to the {\em size} of maximum matching in $\poly(\log n)$ time, where $c$ can be any arbitrarily large constant. Since there are already highly efficient $2$-approximate algorithms, we will use them for this purpose but emphasize that we can instead use any other $O(1)$-approximate algorithm.

\begin{proposition}[See \cite{BernsteinFH-SODA19} or \cite{BehnezhadDHSS-FOCS19}]\label{prop:dynamic-maximal-matching}
	There is a randomized algorithm that maintains a maximal matching of an $n$-vertex fully dynamic graph against an oblivious adversary in $\poly(\log n)$ worst-case update-time.
\end{proposition}

We also use the following algorithm to argue that if the maximum matching of $G$ becomes small at any point during the updates, then there is already an algorithm that can efficiently maintain a $(1-\epsilon)$-approximation during those updates. We use this algorithm to assume that at all times $\mu(G)$ is larger than, say $\log n$, which is useful for our probabilistic events. See \Cref{rem:mularge}.

\begin{proposition}[\cite{GuptaPeng-FOCS13}]\label{prop:alg-small-matching}
	Let $\mu' = \Omega(1)$ and $0 < \epsilon < 1$ be any parameters. There is a deterministic algorithm that maintains a matching $M$ of a fully dynamic graph $G$ with worst case update-time $O(\mu'/\epsilon^2 + \log^3 n)$ satisfying the following: at any time during the updates where $\mu(G) \leq \mu'$ it also holds that $|M| \geq (1-\epsilon)\mu(G)$.
\end{proposition}
\begin{myproof}[Proof sketch]
	The idea is due to \cite{GuptaPeng-FOCS13}. Consider a graph $G$ and suppose that $C$ is a vertex cover of size $O(\mu(G))$ of $G$. Consider a {\em core subgraph} of $G$ that includes all the edges of $G$ with both endpoints in $C$ and also includes $|C|+1$ arbitrary edges of each vertex in $C$. It is not hard to see that a core graph includes a maximum matching of $G$ --- see \cite[Section~3]{GuptaPeng-FOCS13} for the proof.
	
	Suppose that we maintain a $3$-approximate vertex cover $C$ at all times. This can be done deterministically in $O(\log^3 n)$ worst-case update time using the algorithm of \cite{BhattacharyaHN-SODA17}. If at any point during the updates $|C| > 12\mu'$ then we know $\mu(G) \geq |C|/6 > 2\mu'$ so even returning the empty matching satisfies the proposition. Otherwise, we can construct the core graph in $O(|C^2|) = O(\mu'^2)$ time, find a maximum matching of it in $O(\mu'^2/\epsilon)$ time using \Cref{prop:apxmatching} and amortize the cost over the next $\epsilon \mu'$ updates where the maximum matching size can only change by a small amount. As a result, we get a $(1-\epsilon)$-approximation with $O(\mu'/\epsilon^2 + \log^3n )$ amortized update-time. The update time can also be made worst-case by standard techniques. See \cite{GuptaPeng-FOCS13} for more details.
\end{myproof}

The following algorithm also follows from \cite{GuptaPeng-FOCS13} which is helpful when $\Delta$ is small.

\begin{proposition}[\cite{GuptaPeng-FOCS13}]\label{prop:lowdegree}
	Let $\Delta$ be an upper bound on the maximum degree of a fully dynamic graph $G$ at all times. For any $\epsilon > 0$, one can maintain a $(1-\epsilon)$-approximate maximum matching of $G$ in worst-case update-time $O(\Delta/\epsilon^2)$.
\end{proposition}

\section{Hierarchical Edge-Degree Constrained Subgraphs (HEDCS)}\label{sec:HEDCS}
\label{sec:HEDCS}

In this section, we focus on HEDCS, give a few useful definitions for it, and prove some of its key properties. The dynamic algorithm for maintaining an HEDCS is then presented in \Cref{sec:dynamicalg}.

\subsection{Basic Properties of HEDCS}

One useful property of any \HEDCS{\beta}{k} is that its maximum degree is $\leq \beta - 1$, regardless of the value of $k$. This sparsity of HEDCS is particularly useful for maintaining it.

\begin{observation}\label{obs:HEDCS-max-degree}
	Every \HEDCS{\beta}{k} $H$ has maximum degree at most $\beta - 1$. 
\end{observation}
\begin{myproof}
	Fix a hierarchical decomposition $H_1 \subseteq \ldots \subseteq H_k$ of $H$. We define the level of any edge $e \in H$ to be the unique value of $i \in [k]$ such that $e \in H_i \setminus H_{i-1}$.
	
	Towards contradiction suppose $\deg_H(v) \geq \beta$ for some vertex $v$. Take an arbitrary edge $e$ of $v$ with the highest level. Suppose that the level of $e$ is $i$, i.e., $e \in H_i \setminus H_{i-1}$. It holds that $\deg_{H_i}(e) \geq \deg_{H_i}(v) + 1 = \deg_H(v) + 1 \geq \beta + 1$ contradicting property $(i)$ of HEDCS. 
\end{myproof}

While we do not use the next two simple observations in our proofs, it might be instructive to state them here regardless.

\begin{observation}\label{obs:HEDCSkk'}
	For any integers $k' \geq k$ and $\beta$, any \HEDCS{\beta}{k} $H$ is also a \HEDCS{\beta}{k'}.
\end{observation}
\begin{myproof}
	Let $H_1, \ldots, H_k$ be a hierarchical decomposition of $H$ and let $H_{k+1} = \emptyset, \ldots, H_{k'} = \emptyset$. It is easy to see that $H_1, \ldots, H_{k'}$ satisfies properties of \HEDCS{\beta}{k'}, thus $H$ is also a \HEDCS{\beta}{k'}.
\end{myproof}

\begin{observation}
	For any parameters $k \geq 1$ and $\beta \geq 2$, any graph $G$ has a \HEDCS{\beta}{k}. 
\end{observation}
\begin{myproof}
	Any graph $G$ is known to have a $\beta$-EDCS for any $\beta \geq 2$ \cite{BernsteinStein15ArXiv}. Since a $\beta$-EDCS is equivalent to a \HEDCS{\beta}{1}, the statement follows from \Cref{obs:HEDCSkk'}.
\end{myproof}

\subsection{Approximation Ratio of HEDCS: Basic Definitions}

We now turn to measuring the approximation ratio guaranteed by an HEDCS, and give a few definitions for this purpose.

For any integers $\beta > \beta^- \geq 1$, and $k \geq 1$ we define a number $f(k, \beta, \beta^-)$ that plays a crucial role in bounding the approximation ratio achieved by a \HEDCS{\beta}{k} (for now think of $\beta^-$ as a number that is very close to $\beta$ but is smaller, we will specify its value soon). In the definition below, by a {\em bipartite} HEDCS we mean an HEDCS defined on a bipartite graph. 

\begin{definition}[\textbf{Function $f$}]\label{def:f}
	For integers $k \geq 1$ and $\beta > \beta^- \geq 1$, let $f(k, \beta, \beta^-)$ be the largest number such that in every bipartite $\HEDCS{\beta}{k}$ with vertex parts $P$ and $Q$ and $\geq \beta^-|P|/2$ edges, 
	$$
	|Q| \geq f(k, \beta, \beta^-) |P|.
	$$
\end{definition}

Based on $f$, we define a function $\alpha$ that is more convenient to use for our approximations:

\begin{definition}[\textbf{Function $\alpha$}]\label{def:alpha}
	For any integers $k \geq 1$ and $\beta > \beta^- \geq 1$ we define
	$$
	\alpha(k, \beta, \beta^-) = \frac{2f(k,  \beta, \beta^-)}{2f(k, \beta, \beta^-) + 1}.
	$$
\end{definition}

Let us now relate HEDCS to the value of function $\alpha$ defined above.

\hiddencomment{I actually think for the second point below, we can show that if $\beta \geq \log(1/\epsilon)/\epsilon^2$, then for every integer $\beta' \geq 1$, it holds that $\mu(H \cup U) \geq (1-\epsilon)\alpha(k, \beta'+2k-1, \beta)$. This is slightly more interesting because the consntant in the factor revealing LP won't affect the value of $\beta'$ and so the update-time will be independent of the analysis. Anyway, this is minor.}

\newcommand{\generalbeta}[1]{\ensuremath{c(#1)^2\log(#1)}}
\newcommand{\propapxstatement}[0]{
Let $H$ and $U$ be subgraphs of a graph $G$, let $\beta \geq 2$ be any integer, and suppose that $H$ is a $\HEDCS{\beta}{k}$ of $G \setminus U$. Then:
	\begin{itemize}
		\item If $G$ is bipartite, then $\mu(H \cup U) \geq \alpha(k, \beta, \beta - 1) \mu(G)$.
		\item If $\beta \geq \generalbeta{\beta' k}$ for some integer $\beta'$ and a sufficiently
large constant $c \geq 1$, then $\mu(H \cup U) \geq \alpha(k, \beta' + 2k-1, \beta')\mu(G)$. This holds even if $G$ is non-bipartite.
	\end{itemize}
}

\begin{proposition}[Approximation guarantee of HEDCS]\label{prop:apx-HEDCS}
	\propapxstatement{}
\end{proposition}

Subgraph $U$ in \Cref{prop:apx-HEDCS} will be important for our particular application. But it would be instructive to let $U = \emptyset$. Doing so, note that we get $\mu(H) \geq \alpha(k, \beta, \beta-1) \mu(G)$ if $H$ is a $\HEDCS{\beta}{k}$ of a bipartite graph $G$. Hence, $H$ is guaranteed to include an $\alpha(k, \beta, \beta-1)$-approximate matching of its base graph $G$ in this case. The same can be done for general graphs, albeit with a slightly different dependence on the parameter $\beta$.

The proof of \Cref{prop:apx-HEDCS} builds on the proof of \cite{AssadiBernsteinSOSA19} that an EDCS obtains a near $2/3$-approximation. We provide the details of the needed modifications in \Cref{sec:prop-apx-HEDCS-proof}.

So it only remains to lower bound the value of function $\alpha(k, \beta, \beta^-)$ for various $k$, $\beta$, and $\beta^-$ to understand the approximation ratio achieved via HEDCS. Let us start with a trivial bound.

\begin{observation}\label{obs:alpha-lt-half}
	For any $k \geq 1$ and any $\beta > \beta^- \geq 1$, $\alpha(k, \beta, \beta^-) \geq \frac{\beta^-/(\beta-1)}{\beta^-/(\beta-1) + 1}$.
\end{observation}
\begin{myproof}
	Let $H$ be a \HEDCS{\beta}{k} with vertex parts $P$ and $Q$ and at least $|H| \geq \beta^-|P|/2$ edges. Since the maximum degree in $H$ is at most $\beta - 1$ by \Cref{obs:HEDCS-max-degree}, we have $|Q| (\beta - 1) \geq |H| \geq \beta^-|P|/2$. Rearranging the terms, we get $|Q| \geq \frac{\beta^-}{2 (\beta-1)} |P|$ and thus $f(k, \beta, \beta^-) \geq \frac{\beta^-}{2(\beta-1)}$. As such, we get
	$$
	\alpha(k, \beta, \beta^-) = \frac{2f(k, \beta, \beta^-)}{2f(k, \beta, \beta^-)+1} \geq \frac{2 \cdot \frac{\beta^-}{2(\beta-1)}}{2 \cdot \frac{\beta^-}{2(\beta-1)}+1} \geq \frac{\beta^-/(\beta-1)}{\beta^-/(\beta-1)+1}.\qedhere
	$$
\end{myproof}

\Cref{obs:alpha-lt-half} implies that $\alpha(k, \beta, \beta-1) \geq 1/2$ for any $\beta > \beta^- \geq 1$ and $k \geq 1$. From \Cref{prop:apx-HEDCS}, we thus get that a $\HEDCS{\beta}{k}$ for every choice of $k \geq 1$ and $\beta \geq 2$ includes an at least $\sfrac{1}{2}$-approximation for bipartite graphs. For general graphs too, \Cref{obs:alpha-lt-half} and \Cref{prop:apx-HEDCS} together imply that by increasing $\beta$, the approximation ratio of a \HEDCS{\beta}{k} can get arbitrarily close to at least $\sfrac{1}{2}$-approximation.

As we will see, however, much better lower bounds can be proven for $\alpha(k, \beta, \beta-)$ when $k$ is moderately small. In particular, by adapting techniques from \cite{AssadiBernsteinSOSA19} one can show $f(1, \beta, \beta^-)$ gets arbitrarily close to 1 if $\beta$ is large and $\beta^-$ is close to $\beta$. This implies that $\alpha(1, \beta, \beta^-)$ can get arbitrarily close to $2/3$, recovering the $2/3$-approximation guarantee of EDCS.

The analysis of function $\alpha(k, \beta, \beta^-)$, however, is more challenging for $k > 1$. In \Cref{sec:fk-LP} we give an LP-based approach that can lower bound $\alpha(k, \beta, \beta^-)$ for moderately small values of $k$ and $\beta$. We use the bounds achieved by this approach for our claimed approximations for $k \in \{2, 3\}$. Later in \Cref{sec:analysis-of-fkb}, we present a different approach for bounding $\alpha$, implying that for $k = O(1/\epsilon)$ and an appropriate $\beta$, the approximation guarantee is $\frac{1}{2} + \Omega_\epsilon(1)$, i.e., strictly better than half.

\subsection{Approximation Ratio of HEDCS: A Factor Revealing LP}\label{sec:fk-LP}

In this section, we show how to lower bound the value of $f(k, \beta, \beta^-)$ (and thus $\alpha(k, \beta, \beta^-)$) by running a linear program (LP). This LP is particularly useful when the values of $\beta$ and $k$ are not too large. In particular, we use this LP to reveal the approximation factor of our algorithm for $k=2$ and $k=3$.

\newcommand{\Pset}[0]{\ensuremath{\mc{P}}}
\newcommand{\Qset}[0]{\ensuremath{\mc{Q}}}
\newcommand{\ELP}[0]{\ensuremath{E_{LP}}}
\newcommand{\LP}[1]{\ensuremath{LP(#1)}}

The LP is written based on three parameters $\beta$, $\beta^-$, and $k$. Let us start with a number of definitions. Define sets $\Pset = \{0, \ldots, \beta\}^k$ and $\Qset = \{0, \ldots, \beta\}^k$. For any $p \in \Pset$ (resp. $q \in \Qset$) and any $i \in [k]$, we use $p_i$ (resp. $q_i$) to denote the $i$-th entry of $p$ (resp. $q$).  We define $\ELP$ to denote all triplets $(p, q, j) \in \Pset \times \Qset \times [k]$ such that $\sum_{i=1}^j p_i + q_i \leq \beta$.

The LP has four types of variables. First, for any $(p, q, j) \in \ELP$ we have a variable $x(p, q, j)$. Second, for any $p \in \Pset$ we have a variable $n_P(p)$. Third, for any $q \in \Qset$ we have a variable $n_Q(q)$. Fourth, we have a single variable $r$ that is also the objective value to be minimized.

The LP can now be formalized as follows; we use \LP{k, \beta, \beta^-} to denote its optimal value.

\begin{equation*}
\begin{array}{ll@{}ll}
\text{minimize}  & r &\\ \medskip
\text{subject to}&   n_P(p) \cdot p_j = \displaystyle\sum_{q: (p, q, j) \in \ELP} x(p,q,j) \text{\hspace{1cm}} &\text{for all } p \in \Pset \text{ and } j \in [k] \\ \medskip
				 &   n_Q(q) \cdot q_j = \displaystyle\sum_{p: (p, q, j) \in \ELP} x(p,q,j) \text{\hspace{1cm}} &\text{for all } q \in \Qset \text{ and } j \in [k] \\ \medskip
				 &   \sum_{q \in \Qset} n_Q(q) = r\\ \medskip
				 &   \sum_{p \in \Pset} n_P(p) = 1\\ \medskip
				 &   \displaystyle \sum_{(p, q, j) \in \ELP} x(p, q, j) \geq \beta^- /2\\
                 &   x(p, q, j) \geq 0 &\text{for all } (p, q, j) \in \ELP\\ \medskip
                 &   n_Q(q) \geq 0, n_P(p) \geq 0 & \text{for all $p \in \Pset$ and $q \in \Qset$.}
\end{array}
\end{equation*}

In the next lemma, we show that \LP{k, \beta, \beta^-} lower bounds the value of $f(k, \beta, \beta^-)$.

\begin{lemma}\label{lem:LPlowerboundsf}
	For any $k \geq 1$, $\beta$, and $\beta^-$, we have $f(k, \beta, \beta^-) \geq \LP{k, \beta, \beta^-}$.
\end{lemma}

\begin{myproof}
	Suppose that $f(k, \beta, \beta^-) = \rho$. From the definition of $f(k, \beta, \beta^-)$, we get that there exists a bipartite $\HEDCS{\beta}{k}$ $H$ with vertex parts $P$ and $Q$ and at least $\beta^-|P|/2$ edges, such that $|Q| = \rho |P|$. Based on this graph $H$, we construct a feasible solution to the LP, and show that its objective value is $\rho$. This clearly suffices to prove $f(k, \beta, \beta^-) \geq \LP{k, \beta, \beta^-}$.
	
	Let $(H_1, \ldots, H_k)$ be a hierarchical decomposition of $H$ satisfying \HEDCS{\beta}{k} constraints. Let $s = (s_1, \ldots, s_k)$ be a vector with each $s_i \in \{0, \ldots, \beta\}$. We say a vertex $v$ in $H$ has {\em degree-profile} $s$ if for any $i \in [k]$, $\deg_{H_i \setminus H_{i-1}}(v) = s_i$. For any $p \in \Pset$, we use $P(p)$ to denote the subset of vertices in part $P$ of $H$ that have degree-profile $p$. Similarly, for any $q \in \Qset$, we use $Q(q)$ to denote the subset of vertices in part $Q$ of $H$ with degree-profile $q$. For any $p \in \Pset$, $q \in \Qset$, and $j \in [k]$, we use $H(p, q, j)$ to denote the subset of edges of $H$ that belong to $H_j \setminus H_{j-1}$, with the $P$-endpoint having degree-profile $p$ and the $Q$-endpoint having degree profile $q$.
	
	Consider the following values for the variables of the LP:
	\begin{itemize}
		\item For any $p \in \Pset$, we set $n_P(p) = |P(p)|/|P|$.
		\item For any $q \in \Qset$, we set $n_Q(q) = |Q(q)|/|P|$. (Note that the denominator is $|P|$ and not $|Q|$.)
		\item For any $(p, q, j) \in \ELP$ we set $x(p, q, j)$ to be $|H(p, q, j)|/|P|$.
		\item We set $r = |Q|/|P| = \rho$.
	\end{itemize}
	
	Let us now verify that this is a feasible solution for the LP.
	
	For the LHS of the first constraint, we have $n_P(p) \cdot p_j = \frac{|P(p)| p_j}{|P|}$ and for the RHS, we have $\sum_{q: (p, q, j) \in \ELP} |H(p, q, j)|/|P|$. The $|P|$ factors cancel out from both sides, and we just need to show
	$$
		|P(p)| p_j = \sum_{q: (p, q, j) \in \ELP} |H(p, q, j)|.
	$$
	The LHS is the number of vertices in $P$ with degree-profile $p$ times $p_j$. Since every vertex with degree-profile $p$ by definition has exactly $p_j$ edges in $H_j \setminus H_{j-1}$, the LHS counts the number of edges of $H_j \setminus H_{j-1}$ connected to vertices of $P$ with degree-profile $p$. Note from definition of $H(p, q, j)$ that the RHS counts exactly the same quantity as we sum over all possible degree-profiles in the $Q$-side and thus count all the $H_j \setminus H_{j-1}$ edges with the $P$-endpoint having degree-profile $p$.
	
	The second constraint can be verified to be satisfied in exactly the same way as the first.
	
	For the LHS of the third constraint, we have $\sum_{q \in \Qset} n_Q(q) = \sum_{q \in \Qset} |Q(q)|/|P| = \frac{1}{|P|}\sum_{q \in \Qset} |Q(q)|$. Since every vertex in $Q$ has a unique degree-profile, the sum equals $|Q|$. Hence, the LHS of the third constraint equals $|Q|/|P| = \rho$. Since we set $r = \rho$, the third constraint is also satisfied.
	
	For the LHS of the fourth constraint, we have $\sum_{p \in \Pset} n_P(p) = \sum_{p \in \Pset} |P(p)|/|P| = \frac{1}{|P|} \sum_{p \in \Pset} |P(p)|$. The sum counts the number of vertices in $P$, thus this indeed equals one as required by the LP.
	
	For the fifth constraint, observe that the LHS equals $\frac{1}{|P|}\sum_{(p, q, j) \in \ELP} |H(p, q, j)|$. We claim that $\frac{1}{|P|}\sum_{(p, q, j) \in \ELP} |H(p, q, j)| \geq |H|/|P|$. Combined with our discussion of the first paragraph of the proof that $|H| \geq \beta^-|P|/2$, this suffices to prove that the fifth constraint holds. To prove this claim, take an edge $e \in H$, and let $(p, q, j)$ be such that $e \in H(p, q, j)$. We show that $(p, q, j) \in \ELP$ which means means the sum $\sum_{(p, q, j) \in \ELP} |H(p, q, j)|$ counts each edge of $H$ at least once, proving the claim. From the definition of degree-profiles, it can be confirmed that $\deg_{H_j}(e) = \sum_{i=1}^j p_i + q_i$; now since $e \in H_j \setminus H_{j-1}$, from property $(i)$ of \HEDCS{k}{\beta}, we get that $\deg_{H_j}(e) \leq \beta$. Hence, $\sum_{i=1}^j p_i + q_i \leq \beta$ and thus $(p, q, j) \in \ELP$ by definition of \ELP{}.
	
	Finally, the non-negativity constraints can be easily verified to hold since $P(q), Q(q), H(p, q, j), $ and $P$ are all sets and hence have non-negative size.
\end{myproof}

{
\renewcommand{\arraystretch}{1.3}

\begin{table}
\centering

\begin{tabular}{|c|c|c|c|c|l|}
\hline
\rowcolor[HTML]{EFEFEF} 
$k$ & $\beta$ & $\beta^-$ & $f(k, \beta, \beta^-)$ & $\alpha(k, \beta, \beta^-)$ & Note                                                    \\ \hline
2   & 220     & 217       & $\geq .780$          & $\geq .609$               & Used for $k=2$ and general graphs in \Cref{thm:main}.   \\ \hline
2   & 142     & 141       & $\geq .789$          & $\geq .612$               & Used for $k=2$ and bipartite graphs in \Cref{thm:main}. \\ \hline
3   & 47      & 42        & $\geq .569$          & $\geq .532$               & Used for $k=3$ and general graphs in \Cref{thm:main}.   \\ \hline
3   & 35      & 34        & $\geq .645$          & $\geq .563$               & Used for $k=3$ and bipartite graphs in \Cref{thm:main}. \\ \hline
\end{tabular}

\caption{Lower bounds on the values of $f(k, \beta, \beta^-)$ and $\alpha(k, \beta, \beta^-)$ obtained via \LP{k, \beta, \beta^-}.}
\label{table:LP}
\end{table}

}

From \Cref{lem:LPlowerboundsf}, we get that it suffices to run \LP{k, \beta, \beta^-} to lower bound the value of $f(k, \beta, \beta^-)$, and thus $\alpha(k, \beta, \beta^-)$, which determines the approximation ratio of HEDCS for different parameters. \Cref{table:LP} gives some of these results that we use in our approximation guarantees.\footnote{The code is available upon request.}

\section{Maintaining a Hierarchical EDCS in Fully Dynamic Graphs}\label{sec:dynamicalg}
\label{sec:MainAlg}

In this section, we describe an algorithm that maintains a \HEDCS{\beta}{k} in fully dynamic graphs efficiently (the maintained structure actually deviates slightly from a \HEDCS{\beta}{k} since the updates are handled lazily, but the matching maintained is approximately as large as that of a \HEDCS{\beta}{k}). 

The formal guarantee of the algorithm is as follows.

\begin{highlighttechnical}
\begin{theorem}\label{thm:dynamicalg}
	Let $k \geq 0$ and $\beta \geq 2$ be integers, and let $\epsilon \in (0, \sfrac{1}{12})$. Let $G$ be an $n$-vertex fully dynamic graph with $\Delta$ and $m$ being fixed upper bounds on the maximum degree and the number of edges of $G$. Provided that there are at least $\Omega(m)$ edge updates, there is an algorithm that maintains a matching $M$ of $G$ such that:
	\begin{itemize}[leftmargin=10pt, topsep=5pt]
	\item {\normalfont \textbf{Update-time:}} Each update takes
	$\min\{\Delta^{\frac{1}{k+1}}, m^{\frac{1}{2(k+1)}} \} \cdot \poly(\epsilon^{-1} \beta k \log n)$ time w.h.p.
	
	\item {\normalfont \textbf{Approximation:}} Suppose $k \geq 1$ and let $\alpha$ be as in \Cref{def:alpha}. If graph $G$ is bipartite, then at any point, w.h.p., it holds that
	$$|M| \geq \big(\alpha(k, \beta, \beta-1) - O(\epsilon)\big)\cdot \mu(G).$$ 
	If $G$ is not necessarily bipartite and $\beta$ is such that $\beta \geq \generalbeta{\beta' k}$ for some integer $\beta' \geq 1$ and a large enough constant $c \geq 1$, then at any point, w.h.p.,
	$$|M| \geq \big(\alpha(k, \beta'+2k-1, \beta') - O(\epsilon)\big) \cdot \mu(G).$$
	If $k=0$, then at any point, w.h.p., $|M| \geq (1-O(\epsilon)) \cdot \mu(G)$.
	\end{itemize} 
\end{theorem}
\end{highlighttechnical}

For ease of exposition, we have decided not to optimize the $\poly(\epsilon^{-1} \beta k \log n)$ factor in the update-time guarantee of \Cref{thm:dynamicalg}. Note, however, that for our final claimed bounds in \Cref{table:results} we will only need $\epsilon^{-1} \beta k \log n = \poly(\log n)$ and so the update time is $\Ot(\min\{\Delta^{\frac{1}{k+1}}, m^{\frac{1}{2(k+1)}}\})$.

In \Cref{sec:the-algorithm} we present an algorithm that we show obtains the same guarantees as those stated in \Cref{thm:dynamicalg}, except that its {\em amortized} update-time bound is $\Ot(\Delta^{\frac{1}{k+1}})$. We then show in \Cref{sec:worstcase} how the algorithm can be slightly modified to turn this amortized bound to worst-case using standard techniques. Finally, in \Cref{sec:degree-root-m} we show how to ignore some of the edges of the graph as they are inserted, such that the maximum degree remains $O(\sqrt{m}/\epsilon)$ without changing the maximum matching of the graph by much. From this, we get the claimed $\Ot(\min\{\Delta^{\frac{1}{k+1}}, m^{\frac{1}{2(k+1)}}\})$ worst-case update-time guarantee of \Cref{thm:dynamicalg}.

\subsection{High-Level Overview of the Algorithm}\label{sec:alg-highlevel}

Our starting point is a pre-processing algorithm where we construct three sequences of subgraphs $G_i, U_i, H_i$ of $G$ satisfying certain structures. We maintain these subgraphs upon updates too, but in a lazy fashion. The following properties, in particular, continue to hold throughout the algorithm:
\begin{enumerate}[label=$(\roman*)$,itemsep=0pt,topsep=0pt]
	\item $\emptyset = G_0 \subseteq G_1 \subseteq \ldots \subseteq G_{k+1} = G$,
	\item $G = U_1 \supseteq U_2 \supseteq \ldots \supseteq U_{k+1}$, and
	\item $\emptyset = H_0 \subseteq H_1 \subseteq \ldots \subseteq H_{k}$.
\end{enumerate}

Subgraphs $G_i$ are simply edge-sampled random subgraphs of $G$. Specifically, for each edge $e$ we draw a real $\pi_e \sim \unifzeroone$ independently (upon arrival of the edge) and $e$ appears in $G_i$ iff $\pi_e \leq p_i$ for some parameters $p_1 \leq \ldots \leq p_{k+1} = 1$ (defined in \Cref{alg:preprocessing}).

Each subgraph $H_i$, at any time, will be a \HEDCS{\beta}{i} of graph $(G \setminus G_i) \setminus U_{i+1}$ with $(H_1, \ldots, H_i)$ being its hierarchical decomposition (we prove this in \Cref{lem:it-is-HEDCS}). Subgraph $U_{i+1}$, in particular, will include all $(H_{i}, \beta)$-underfull edges in $G \setminus G_i$ at any time. This immediately proves Property $(ii)$ of \HEDCS{\beta}{i} for $H_i$ since no $(H_i, \beta)$-underfull edge is in $(G \setminus G_i) \setminus U_{i+1}$ (as they are all in $U_{i+1}$). 

The final matching $M$ that we output, and lazily maintain, is an (almost) maximum matching of the edges of $H_k \cup U_{k+1}$ that are present in the graph. By our discussion above, $H_k$ will be a \HEDCS{\beta}{k} of $(G \setminus G_k) \setminus U_{k+1}$, thus by the guarantee of HEDCS (\Cref{prop:apx-HEDCS}), $\mu(H_k \cup U_{k+1})$ should well-approximate $\mu(G \setminus G_k) \approx \mu(G)$ where the last almost-equality comes from the fact that $G_k$ includes $p_k = o(1)$ fraction of the edges of $G$ randomly, hence their removal does not change the matching size by much. It would be useful to note here that during the updates, some of the edges of each subgraph $H_i$ may get removed from the graph. We do not remove these edges from $H_i$ immediately. Rather, we recompute $H_i$ frequently enough to ensure that at any point only a small number of its edges have been removed from the graph. We note that the output matching $M$ will not use the edges of $H_k$ that are deleted and will always be a proper matching of $G$.

\newcommand{\mulb}[0]{\ensuremath{10^3 \epsilon^{-1} k \log n}}

\begin{remark}[Assumptions]\label{rem:mularge}
	We make the following assumption throughout the rest of \Cref{sec:dynamicalg} which all hold without loss of generality (w.l.o.g.).

	We assume that $\mu(G) \geq \mu' = \mulb$ throughout the whole sequence of updates. This assumption comes w.l.o.g. since by plugging this value of $\mu'$ in \Cref{prop:alg-small-matching}, we get an algorithm with worst-case update time $O(\log^3 n + k \log n/\epsilon^3)$ that already maintains a $(1-\epsilon)$-approximate matching of $G$ whenever $\mu(G) < \mu'$. 
	
	We assume $\Delta \geq 15 \log n /\epsilon$ as otherwise \Cref{prop:lowdegree} already gives a $(1-\epsilon)$-approximation with $O(\log n / \epsilon)$ worst-case update time. 
	
	We assume $k$ is small enough that $\Delta^{\frac{1}{k+1}} \geq 15 \log n / \epsilon$. If not, we can pick a smaller $k$ that satisfies it and additionally $\Delta^{\frac{1}{k+1}} = O(\log n / \epsilon)$. Note that by picking a smaller $k$ the approximation improves and the update-time of \Cref{thm:dynamicalg} would be $\poly(\epsilon^{-1}\beta k \log n)$ when $\Delta^{\frac{1}{k+1}} = O(\log n / \epsilon)$.
\end{remark}

\vspace{-0.7cm}
\subsection{The Formal Algorithm}\label{sec:the-algorithm}
\vspace{-0.2cm}

We start with the pre-processing algorithm described as \Cref{alg:preprocessing}.

\begin{algorithm}[H]
	\textbf{Input:} A graph $G=(V, E)$ with maximum degree bounded by $\Delta$.
	
	\textbf{Parameters:} Integers $\beta \geq 1$ and $k \geq 1$ and real $\epsilon \in (0, 1)$.
	
	For any $i \in [k]$ let $p_i := \epsilon \cdot \Delta^{\frac{i}{k+1}-1}$ and let $p_{k+1} = 1$. \Comment{See \Cref{cl:pvalues}.}
	
	Draw a random ranking $\pi$, i.e., for each edge $e \in E$ draw $\pi_e \sim \unifzeroone$ independently.
	
	For any $i \in [k+1]$ let $G_i$ be the subgraph of $G$ including all edges $e$ with $\pi_e \leq p_i$.
	
	Let $U_1 \gets G$, let $G_0 \gets \emptyset$, and let $H_0 \gets \emptyset$.
	
	Run $\computeLayers{1}$ (formalized as \Cref{alg:computeLayers}) to generate subgraphs $H_1, \ldots, H_{k+1}$, $U_2, \ldots, U_{k+1}$, integers $\mu_1, \ldots, \mu_{k+1}$,  and the output matching $M$.
	\caption{$\PreProcessing{G}$}
	\label{alg:preprocessing}	
\end{algorithm}

Subroutine \computeLayers{j} formalized below in \Cref{alg:computeLayers} is called both in pre-processing \Cref{alg:preprocessing} (for $j=1$) and during the update time (for various $j \in [k+1]$). 

\begin{algorithm}[H]	
	\For{$i$ in $j, \ldots, k$}{	
		$\mu_i \gets \maximalmatchingsize{G_i}$ \Comment{See paragraph ``Maintaining $\maximalmatchingsize{G_i}$'' below.}
		
		$H_i \gets \addLayer{U_i \cap G_i, H_{i-1}, \mu_i}$. \Comment{Formalized as \Cref{alg:H}.}

		Let $U_{i+1}$ be the graph including every edge in $U_i \setminus G_i$ that is $(H_i, \beta)$-underfull.
	}

	$\mu_{k+1} \gets \maximalmatchingsize{G_{k+1}}$ \Comment{See paragraph ``Maintaining $\maximalmatchingsize{G_i}$'' below.}
	
	$M \gets$ a $(1-\epsilon)$-approximate max matching of $(H_{k} \cup U_{k+1}) \cap G$ computed via \Cref{prop:apxmatching}.
	
	\caption{$\computeLayers{j}$.}
	\label{alg:computeLayers}
\end{algorithm}

The next subroutine \addLayer{\Gamma, H_{i-1}, \mu_i}, which is called only from \Cref{alg:computeLayers}, starts with $H_i \gets H_{i-1}$, then iterates over the edges of $\Gamma$ in the increasing order of ranks $\pi$ (the same rank function as in \Cref{alg:preprocessing}), adds each encountered underfull edge to to $H_i$ (and removes their incident overfull edges that belong to $H_i \setminus H_{i-1}$,  if any). After the algorithm iterates over sufficiently many edges without detecting any underfull edges, it returns $H_i$.

{
\begin{algorithm}[H]	
	Let $H_i \gets H_{i-1}$, $\eta \gets 0$.
	
	\SetKwBlock{Fna}{\textnormal{Iterate over the edges of $\Gamma$ in the increasing order of $\pi$. Upon visiting an edge $e=(u, v)$:\label{line:loopalgH}}}{}
	\Fna{
	
	$\eta \gets \eta+1$.
	
	\If{$e$ is $(H_{i}, \beta)$-underfull\label{line:proc-underfull-condition}}{
		Add $e$ to $H_i$.\label{line:addToH-algH}
		
		If there exists any $(H_i, \beta)$-overfull edge $(u, w) \in H_i \setminus H_{i-1}$ remove one arbitrarily.

		If there exists any $(H_i, \beta)$-overfull edge $(v, w) \in H_i \setminus H_{i-1}$ remove one arbitrarily.
		
		$\eta \gets 0$.
	}
	
	\If{$\eta > \lfloor |\Gamma|/(4 \mu_i \beta^2 + 1) \rfloor$\label{line:termination-condition}}{
		\smallskip
		\Return $H_i$. \Comment{\Cref{cor:algreturnsH} guarantees that we reach this line eventually.}\label{line:returnH}
	}
	}
	\caption{$\addLayer{\Gamma, H_{i-1}, \mu_i}$. Input $\Gamma$ will be $U_i \cap G_i$ when called.}
	\label{alg:H}
\end{algorithm}
}

\paragraph{Handling edge updates:} We would like to maintain the same output as that of \Cref{alg:preprocessing}. However, to optimize the update-time we handle most edge updates in a lazy fashion. Only the following trivial updates are done immediately upon insertion/deletion:
\begin{itemize}
	\item Upon insertion of an edge $e$ to $G$, we add $e$ to any graph $U_{i+1}$ where $e$ is $(H_{i}, \beta)$-underfull and $e \not\in G_i$. Moreover, we immediately draw the rank $\pi_e \sim \unifzeroone$ for $e$ and for any $i \in [k+1]$ with $\pi_e \leq p_i$, we add $e$ to graph $G_i$.
	\item Upon deletion of an edge $e$ from $G$, we immediately remove $e$ from any of $G_1, \ldots, G_{k+1}$, $U_1, \ldots, U_{k+1}, M$ that includes $e$. Note that we do {\em not} remove $e$ from $H_1, \ldots, H_k$.
\end{itemize}
By storing all the graphs in the adjacency-list format and storing each adjacency-list as a balanced binary search tree, these operations can easily be implemented in $O(k \log n)$ time per update.

As discussed, the more time-consuming updates are done lazily. Particularly, for any $i \in [k + 1]$ we keep a counter $c_i$ that is initially zero after the pre-processing step. Then upon every update we set $c_i \gets c_i + 1$ for all $i \in [k + 1]$. For any $j \in [k+1]$ immediately after condition $c_j \geq \frac{\epsilon}{k} \cdot \frac{\mu_j + 1}{p_j}$ holds, we reset $c_j, c_{j+1}, \ldots, c_{k+1}$ to zero and call subroutine $\computeLayers{j}$ of \Cref{alg:computeLayers} to recompute $H_j, \ldots, H_k$, $U_{j+1}, \ldots, U_{k+1}$, and $M_j, \ldots, M_{k+1}$. (If for multiple $j$ the condition $c_j \geq \frac{\epsilon}{k} \cdot \frac{\mu_j + 1}{p_j}$ holds at once, we apply the procedure on the minimum such $j$.)

\paragraph{Maintaining $\maximalmatchingsize{G_i}$:} In \Cref{alg:computeLayers} we need to compute the size of a maximal matching $\maximalmatchingsize{G_i}$ for any $i \in \{j, \ldots, k+1\}$. This can of course be done in time $|G_i|$ for each $i$ by iterating over the edges of $G_i$ and constructing a maximal matching greedily. However, we need a more efficient algorithm. To do this, for each $i \in [k+1]$ we maintain a maximal matching of $G_i$ after each edge update along with its size. This can be done in $\poly(\log n)$ worst-case update-time for each $G_i$ using \Cref{prop:dynamic-maximal-matching}. Hence, whenever we call \Cref{alg:computeLayers} we have $\maximalmatchingsize{G_i}$ readily computed. That is, it can be accessed in $O(1)$ time, without the need to go over all edges of $G_i$ to construct the maximal matching. The cost is only an additive worst-case $k \cdot \poly(\log n)$ factor to the final update-time. It is worth noting that instead of $\maximalmatchingsize{G_i}$, any $O(1)$-approximation of $\mu(G_i)$ would suffice for our purpose and the maximality is not important. We decided to use maximal matching algorithms since they are already fast enough.

In \Cref{sec:basic-props} we state some basic properties of the algorithm described above. In \Cref{sec:update-time} we analyze the update-time of this algorithm. Then in \Cref{sec:approx-fk} we analyze its approximation.

\subsection{Basic Properties of the Algorithm of \Cref{sec:the-algorithm}}\label{sec:basic-props}

In this section we prove a number of basic properties of the algorithm that we later use to analyze its update-time and approximation.

\begin{claim}\label{cl:pvalues}
	Values of $p_1, \ldots, p_{k+1}$ are set in \Cref{alg:preprocessing} such that they satisfy the following:
	\begin{enumerate}[label=$(\roman*)$, itemsep=0pt, topsep=0pt]
		\item $p_1 \leq p_2 \leq \ldots \leq p_{k+1} = 1$.
		\item $p_1 \geq 15 \frac{\log n}{\Delta}$.
		\item $p_i/p_{i-1} = O(\Delta^{\frac{1}{k+1}} / \epsilon)$ for all $i \geq 2$.
		\item $p_k \leq \epsilon$.
		\item $\frac{p_i - p_{i-1}}{1 - p_{i-1}} \geq p_i/2$ for all $i \geq 2$.
	\end{enumerate}
\end{claim}
\begin{proof}
	Note from \Cref{alg:preprocessing} that $p_i := \epsilon \cdot \Delta^{\frac{i}{k+1}-1}$ for $i \in [k]$ and $p_{k+1} = 1$. From this definition, we immediately get $p_1 \leq \ldots \leq p_k \leq \epsilon < p_{k+1} = 1$, hence proving part $(i)$. 
	
	Property $(ii)$ follows since $p_1 = \epsilon \cdot \Delta^{\frac{1}{k+1} -1}$ and we assumed in \Cref{rem:mularge} that $\Delta^{\frac{1}{k+1}} \geq 15 \log n / \epsilon$.
	
	For Property $(iii)$, note that $p_i / p_{i-1} \leq \Delta^{\frac{i}{k+1}-1} / (\epsilon \Delta^{\frac{i-1}{k+1} - 1}) = O(\Delta^{\frac{1}{k+1}}/\epsilon).$
	
	Property $(iv)$ trivially holds since $p_k = \epsilon \Delta^{\frac{k}{k+1} -1} \leq \epsilon$.
	
	Property $(v)$ holds since $\frac{p_i - p_{i-1}}{1 - p_{i-1}} \geq p_i - p_{i-1} \geq p_i - p_i/\Delta^{\frac{1}{k+1}} \geq p_i/2$. (The latter holds so long as $\Delta^{\frac{1}{k+1}} \geq 2$ and recall that in \Cref{rem:mularge} we assume it is indeed much larger.)
\end{proof}

\begin{claim}\label{obs:monotonicity}
	It holds at all times that:
	\begin{enumerate}[label=$(\roman*)$, itemsep=0pt, topsep=0pt]
		\item $G_1 \subseteq G_2 \subseteq \ldots \subseteq G_{k+1} = G$.
		\item $\emptyset = H_0 \subseteq H_1 \subseteq H_2 \subseteq \ldots \subseteq H_k$.
		\item $G = U_1 \supseteq U_2 \supseteq \ldots \supseteq U_{k+1}$.
		\item For any $i \in [k]$, $U_{i+1}$ is the set of all $(H_i, \beta)$-underfull edges in $G \setminus G_i$.
		\item $H_i \subseteq G_i$ for all $i \in \{0, \ldots, k\}$.
	\end{enumerate}
\end{claim}
\begin{myproof}
	Property $(i)$ follows from \Cref{cl:pvalues} part $(i)$ since $G_i$ simply includes an edge $e$ iff $\pi_e \leq p_i$, at all times.
	
	For Property $(ii)$, observe that $H_i$'s are only modified in \Cref{alg:computeLayers}. Particularly $H_i$ is obtained by calling \addLayer{U_i \cap G_i, H_{i-1}, \mu_i} which only adds some of the edges of $U_i \cap G_i$ to $H_{i-1}$. Additionally, any time that $H_i$ is recomputed, all of $H_i, \ldots, H_k$ are also recomputed. Hence, the property continues to hold at all times.
	
	We prove Property $(iv)$ by induction. The base case $i=0$ holds trivially since $U_1$ is always equal to $G$. Now observe that we set $U_{i+1}$ in \Cref{alg:computeLayers} to be the graph including $(H_i, \beta)$-underfull edges in $U_i \setminus G_i$. By induction hypothesis, $U_i$ is the subgraph of $(H_i, \beta)$-underfull edges in $G \setminus G_{i-1}$. Hence, $U_{i+1}$ includes every $(H_i, \beta)$ underfull edge in $(G \setminus G_{i-1}) \setminus G_i = G \setminus G_i$ (by Property $(i)$). During the update time, we maintain the invariant that $U_{i+1}$ is the set of all $(H_i, \beta)$-underfull edges not in $G_i$ for every update. Hence, the property holds at all times.
		
	For Property $(iii)$, note that if $e \in U_{i+1}$, then it must be $(H_{i}, \beta)$-underfull and in $G \setminus G_i$ by Property $(iv)$. Since $H_{i-1} \subseteq H_{i}$ by Property $(ii)$, then $e$ is also $(H_{i-1}, \beta)$-underfull and definitely in $G \setminus G_{i-1}$ (since $G_{i-1} \subseteq G_i$ by Property $(i)$). This means $e$ should also belong to $U_i$.
	
	We prove Property $(v)$ by induction on $i$. The base case $i = 0$ trivially holds since $H_0 = \emptyset$. Now observe that subgraph $H_i$ is obtained in \Cref{alg:H} by adding some edges of $U_i \cap G_i$ to $H_{i-1}$. This means $H_i \subseteq (G_i \cap U_i) \cup H_{i-1} \subseteq G_i \cup H_{i-1} \subseteq G_i \cup G_{i-1} = G_i$ where the last two follow from the induction hypothesis and Property $(i)$ respectively.
\end{myproof}

The next observation will play an important role later in \Cref{lem:it-is-HEDCS} where we argue that $H_k$ is a \HEDCS{\beta}{k} of $(G \setminus G_k) \setminus G_{k+1}$.

\begin{observation}\label{obs:H-no-overfull}
	By the end of every iteration of the loop in \Cref{line:loopalgH} of \Cref{alg:H} (and thus by the end of the whole algorithm also), the subgraph $H_i \setminus H_{i-1}$ includes no $(H_i, \beta)$-overfull edge.
\end{observation}
\begin{myproof}
	We prove by induction on the number of iterations. Initially $H_i \setminus H_{i-1}$ is empty and so the claim holds. Once we add an edge $e$ to $H_i$, $e$ must be $(H_i, \beta)$-underfull meaning that $\deg_{H_i}(e) < \beta - 1$. Hence, after inserting $e$ to $H_i$, we have $\deg_{H_i}(e) < \beta + 1$ which means $e$ is not $(H_i, \beta)$-overfull. Note, however, that adding $e$ to $H_i$ increases its endpoints' degrees in $H_i$ by one. This may lead to $(H_i, \beta)$-overfull edges in $H_i \setminus H_{i-1}$ that are incident to $e$, but removing any one such edge from each endpoint of $e$ ensures that $H_i \setminus H_{i-1}$ remains to have no $(H_i, \beta)$-overfull edges.
\end{myproof}

We will also use the following upper bound on the maximum degree of any $H_i$ several times.

\begin{observation}\label{obs:degree-of-H}
	For any $i \in \{0, \ldots, k\}$, the maximum degree of $H_i$ is at most $\beta$ at all times.
\end{observation}
\begin{myproof}
	We prove by induction on $i$. For the base case $H_0 = \emptyset$ and so the claim clearly holds. Suppose, now, that the claim holds for $H_{i-1}$, we prove it for $H_i$. Consider a call to \Cref{alg:computeLayers} where we set $H_i \gets \addLayer{U_i \cap G_i, H_{i-1}, \mu_i}$. It suffices to prove that at this point, $H_i$ has maximum degree $\beta$ since until $H_i$ gets re-computed, we may only remove edges from it.
	
	Suppose toward contradiction that $\deg_{H_i}(v) > \beta$ for some vertex $v$. Since $\deg_{H_{i-1}}(v) \leq \beta$ by the induction hypothesis and $H_{i-1} \subseteq H_i$ by \Cref{obs:monotonicity}, $v$ must have an edge $e=(v, u) \in H_i \setminus H_{i-1}$. Moreover, $\deg_{H_i}(e) = \deg_{H_i}(u) + \deg_{H_i}(v) \geq \deg_{H_i}(v) \geq \beta + 1$ and so $e$ must be $(H_i, \beta)$-overfull. This, however, contradicts \Cref{obs:H-no-overfull}. Hence $\deg_{H_i}(v) \leq \beta$ for all $v$.
\end{myproof}

The following claim bounds the number of times that we encounter an $(H_i, \beta)$-underfull edge in \Cref{alg:H} by $4 \mu_i \beta^2$. The proof is based on a potential function used previously for EDCS by \cite{BernsteinSteinSODA16,AssadiBBMSSODA19,AssadiBernsteinSOSA19,BernsteinICALP20} with a simple additional idea that bounds the number of edges of $H_i$ by $2\mu_i \beta$. For completeness, we provide the full proof in \Cref{sec:proofs}.

\begin{claim}\label{cl:potential}
	\Cref{alg:H} reaches \Cref{line:addToH-algH} at most  $4\mu_i\beta^2$ times.
\end{claim}

As an immediate corollary of \Cref{cl:potential} we get that \Cref{alg:H} always terminates:

\begin{corollary}[of \Cref{cl:potential}]\label{cor:algreturnsH}
	\Cref{alg:H} reaches \Cref{line:returnH} with probability one.
\end{corollary}
\begin{myproof}
	Suppose for the sake of contradiction that the algorithm does not reach \Cref{line:returnH}. Note from \Cref{alg:H} that this means $\eta \leq \tau$ at all times where $\tau = \lfloor |\Gamma|/(4\mu_i\beta^2+1) \rfloor$. From the definition of counter $\eta$ in \Cref{alg:H}, we get that the algorithm encounters at least one $(H_i, \beta)$-underfull edge within every $\tau$ consecutive edges of $\Gamma$ (processed in the order of $\pi)$. But this means, we must encounter at least $|\Gamma|/\tau \geq 4\mu_i \beta^2+1$ edges that are $(H_i, \beta)$-underfull, contradicting \Cref{cl:potential}.
\end{myproof}

\subsection{The Approximation Ratio}\label{sec:approx-fk}

We now analyze the size of $M$ and prove the approximation guarantee of \Cref{thm:dynamicalg}.

Fix an arbitrary sequence of updates and suppose that we run the algorithm of \Cref{sec:the-algorithm} on them. Unless otherwise stated, when we refer to a data structure of the algorithm throughout \Cref{sec:approx-fk} (such as matching $M$, subgraph $H_i$, integer $\mu_i$, graph $G$, etc.) we refer to the value stored in this data structure after the whole sequence of updates has been applied. 

We prove that, w.h.p., the size of $M$ is as claimed in \Cref{thm:dynamicalg} at the end of applying this sequence of updates. Note that since this holds for any arbitrary sequence, it also holds for any update throughout the sequence.

Our starting point is the following lemma.

\begin{lemma}\label{lem:it-is-HEDCS}
	$H_k$ is a \HEDCS{\beta}{k} of $(G \setminus G_k) \setminus U_{k+1}$ with $(H_1, \ldots, H_k)$ being its hierarchical decomposition.
\end{lemma}
\begin{myproof}
	Let us first confirm property $(i)$ of \HEDCS{\beta}{k}. Observe that $H_i$ is computed in \Cref{alg:computeLayers} and is the output of \addLayer{U_i \cap G_i, H_{i-1}, \mu_i}. As such, by \Cref{obs:H-no-overfull}, $H_i \setminus H_{i-1}$ includes no $(H_i, \beta)$-overfull edges. That is, for any edge $e \in H_i$, we have $\deg_{H_i}(e) \leq \beta$ right after recomputation of $H_i$. Now note that any time that some subgraph $H_j$ for $j \leq i$ is recomputed, $H_i$ gets recomputed too in \Cref{alg:computeLayers}. Moreover, other than these recomputations, the subgraphs $H_i$ do not change (even if their edges are removed from the graph). Hence, property $(i)$ of \HEDCS{\beta}{k} continues to hold throughout the sequence of updates.
	
	For property $(ii)$ of \HEDCS{\beta}{k}, recall from \Cref{obs:monotonicity} part $(iv)$ that all $(H_k, \beta)$-underfull edges of $G \setminus G_k$ belong to $U_{k+1}$ at all times. Hence, $(G \setminus G_k) \setminus U_{k+1}$ includes no $(H_k, \beta)$-underfull edges. That is, for any $e \in (G \setminus G_k) \setminus U_{k+1}$, we have $\deg_{H_k}(e) \geq \beta -1$.
\end{myproof}

The only remaining problem is that some of the edges of $H_k$ may have been deleted from the graph. Therefore, although $H_k$ is a valid $\HEDCS{\beta}{k}$ of $G \setminus U_{k+1}$, not all the edges in $H_k \cup U_{k+1}$ actually exist in the graph. The rest of this section is essentially devoted to upper bounding the number of such deleted edges.

A few definitions are in order. We use $\hat{G_i}$ to denote subgraph $G_i$ right after the last recomputation of $\mu_i$ and $H_i$, which must have been exactly $c_i$ updates ago. We also use $F_i$ to refer to any edge inserted or removed at least once from $G_i$ during the last $c_i$ updates. 

A key claim to bounding the number of deleted edges of $H_k$, say, is the following:

\begin{lemma}\label{lem:sum-Di-small}
	With probability $1-1/n^4$, $|F_1| + \ldots + |F_{k+1}| \leq 13\epsilon \mu(G).$
\end{lemma}

Let us first see why \Cref{lem:sum-Di-small} proves the approximation guarantee of \Cref{thm:dynamicalg}.

\begin{proof}[Proof of approximation guarantee of \Cref{thm:dynamicalg}]
	Consider the last time that we recomputed $M$; this must have been $c_{k+1}$ updates ago. Let us use $H'_i$, $U'_i$, and $G'$ to denote $H_i$, $U_i$, and $G$ right before recomputation of $M$. Observe that we set $M$ to be a $(1-\epsilon)$-approximate maximum matching of $(H'_k \cup U'_{k+1}) \cap G'$. We show that
	\begin{flalign}
		\nonumber |M| &\geq (1-\epsilon)\mu\big((H'_k \cup U'_{k+1}) \cap G'\big) - |G' \setminus G|\\
		&\geq (1-\epsilon)\mu\big((H_k \cup U_{k+1}) \cap G \big) - |G' \setminus G| - |G \setminus G'| - |U_{k+1} \setminus U'_{k+1}| - |H_k \setminus H'_k|.\label{eq:tllrcgg-12309}
	\end{flalign}
	The first inequality holds since $M$ at the end of the sequence is its last value computed, which is a $(1-\epsilon)$-approximate maximum matching of $(H'_k \cup U'_{k+1}) \cap G'$, excluding its edges that have been removed from the graph during the last $c_{k+1}$ updates. The second bound holds sincy by replacing $G'$ with $G$, $H'_k$ with $H_k$, and $U'_{k+1}$ with $U_{k+1}$, we may only add $(G \setminus G') \cup (H_k \setminus H'_k) \cup (U_{k+1} \setminus U'_{k+1})$ edges to the matching, which is then canceled out with the subtracted terms.
	
	Now note that $H'_k = H_k$ where recall we use $H_k$ to denote the final value of $H_k$ at the end of the update sequence. This is correct because $M$ gets recomputed any time that \computeLayers{j} is called for {\em any} value of $j$. Hence, since the last recomputation of $M$, we have not called \computeLayers{j} for any $j$ and so $H_k$ must have remained unchanged. Note also that any edge in $U_{k+1} \setminus U'_{k+1}$, $G \setminus G'$, or $G' \setminus G$ must have been updated during the last $c_{k+1}$ updates. Since $F_{k+1}$ by definition includes any edge of $G_{k+1} = G$ updated during the last $c_{k+1}$ updates, (\ref{eq:tllrcgg-12309}) gives
	\begin{flalign}\label{eq:tcclrxxx-1239}
		|M| &\geq (1-\epsilon)\mu\big((H_k \cup U_{k+1}) \cap G\big) - 3|F_{k+1}|.
	\end{flalign}
	
	Next, observe that $U_{k+1} \subseteq G$ by \Cref{obs:monotonicity}. Moreover, if an edge $e \in H_i \setminus H_{i-1}$ does not belong to $G$, then $e$ must have been updated during the last $c_i$ updates and must belong to $G_i$ which together imply $e \in F_i$. Combined with $H_1 \subseteq \ldots \subseteq H_k$ of \Cref{obs:monotonicity}, we get $|H_k \setminus G| \leq |F_1| + \ldots + |F_k|$. Now combined with (\ref{eq:tcclrxxx-1239}) this implies
	\begin{flalign}
		&& \nonumber |M| &\geq (1-\epsilon)\mu\big(H_k \cup U_{k+1}\big) - 3|F_{k+1}| - (|F_1| + \ldots + |F_k|)\\
		&& \nonumber &\geq (1-\epsilon)\mu\big(H_k \cup U_{k+1}\big) - 3(|F_1| + \ldots + |F_{k+1}|)\\
		&& &\geq (1-\epsilon)\mu\big(H_k \cup U_{k+1}\big) - 39 \epsilon \mu(G). && (\text{Holds w.h.p. by \Cref{lem:sum-Di-small}.})\label{eq:clllhbkxx-123}
	\end{flalign}
	
	This immediately proves the approximation guarantee of \Cref{thm:dynamicalg} for $k = 0$ since at all times $U_{k+1} = U_1 = G$. So let us now focus on $k \geq 1$.
	
	Observe that by \Cref{lem:it-is-HEDCS} $H_k$ is a $\HEDCS{\beta}{k}$ of $(G \setminus G_k) \setminus U_{k+1}$ and so applying \Cref{prop:apx-HEDCS} and noting from the statement of \Cref{thm:dynamicalg} that $\beta = \generalbeta{\beta'}$, we get 
	\begin{equation}\label{eq:ccrrcg-123}
		\mu(H_k \cup U_{k+1}) \geq \alpha(k, \beta, \beta-1) \cdot \mu(G \setminus G_k) \qquad \text{ for bipartite $G$},
	\end{equation}
	\begin{equation}\label{eq:llcgg3-123}
		 \mu(H_k \cup U_{k+1}) \geq \alpha(k, \beta', \beta'-k) \cdot \mu(G \setminus G_k) \qquad \text{for general $G$.}
	\end{equation}
	
	To complete the proof, note that $G_k$ includes each edge of $G$ independently with probability $p_k \leq \epsilon$ by \Cref{cl:pvalues} part $(v)$. This means that fixing a maximum matching of $G$, only $\epsilon$ fraction of its edges appear in $G_k$ in expectation. This bound also holds with high probability by a Chernoff bound noting that $\mu(G) \geq \mulb$ from \Cref{rem:mularge}. As such, we get that with probability, say, $1-1/n^4$, $\mu(G \setminus G_k) \geq (1-\epsilon)\mu(G)$. Plugging this into (\ref{eq:ccrrcg-123}) for bipartite graphs and  (\ref{eq:llcgg3-123}) for general graphs, and then applying (\ref{eq:clllhbkxx-123}), we get the claimed lower bound of \Cref{thm:dynamicalg} on $|M|$.
\end{proof}

\hiddencomment{Minor detail: $H_k$ may not be a subgraph of $G$ but the current statement of \Cref{prop:apx-HEDCS} assumes it is. Either change the proposition, our take the union with $H_k$.}

Toward proving \Cref{lem:sum-Di-small}, we prove two auxiliary claims first.

\begin{claim}\label{cl:clrr-12390123}
	It holds that $\mu(G) \geq \frac{2}{3(k+1)}(\mu_1 + \ldots + \mu_{k+1}) - \sum_{i=1}^{k+1} |F_i|$ with probability 1.
\end{claim}
\begin{myproof}
	Recall that $\mu_i$ is the size of a maximal matching of $\hat{G}_i$. Let us denote this maximal matching by $M_i$. Let us also use $M_{1, k+1}$ to denote $M_1 \cup \ldots \cup M_{k+1}$. Observe that if an edge of $M_i$ does not belong to $G$, then it must be in $F_i$. This means $\mu(G) \geq \mu(M_{1, k+1}) - \sum_{i=1}^{k+1} |F_i|$.
	
	Let us now lower bound $\mu(M_{1, k+1})$. To any edge $e \in M_{1, k+1}$ we assign fractional value $x_e := h_e/(k+1)$ where $h_e$ is the number of matchings $M_1, \ldots, M_{k+1}$ that include $e$. It can be confirmed that $x$ is a valid fractional matching of $M_{1, k+1}$. On the other hand, this fractional matching has size exactly $\frac{|M_1| + \ldots + |M_{k+1}|}{k+1} = \frac{\mu_1 + \ldots + \mu_{k+1}}{k+1}$. Any general graph has an integral matching of size at least $2/3$ times the size of any of its fractional matchings. Hence, $\mu(M_{1, k+1}) \geq \frac{2}{3(k+1)} (\mu_1 + \ldots + \mu_{k+1})$. Combined with the bound of the previous paragraph, we thus get
	
	$$
	\mu(G) \geq \mu(M_{1, k+1}) - \sum_{i=1}^{k+1} |F_i| \geq \frac{2(\mu_1 + \ldots + \mu_{k+1})}{3(k+1)} -  \sum_{i=1}^{k+1} |F_i|.\qedhere
	$$
\end{myproof}

Next, we prove the following high probability upper bound on $|F_i|$.

\begin{claim}\label{cl:cclcrg128391t-1h2t3}
	With probability $1-1/n^4$, it holds that $|F_i| <  2\frac{\epsilon}{k} \mu_i + 200 \log n + 4$ for all $i \in [k+1]$.
\end{claim}
\begin{myproof}
	Fix some integer $t \geq 1$. Since the update sequence is oblivious to the randomization of the algorithm, we expect exactly $p_i \cdot t$ edges of the last $t$ updates to have rank $\leq p_i$, i.e., belong to $G_i$. Applying Chernoff and union bounds, we get that with probability $1-1/n^4$, for any $t \leq n^2$ and any $i \in [k]$ (recall that $k \leq \log n$), we have at most $p_i \cdot t + \sqrt{100 (p_i \cdot t) \log n}$ updates to $G_i$ among the last $t$ updates in the sequence (note that this holds for all values of $p_i \cdot t$). We assume this high probability event holds and proceed to prove the claim.

	By definition of $F_i$, any edge in $F_i$ must have been updated in the last $c_i$ updates, and that each of these edges must have belonged to $G_i$. By the discussion above, at most $p_i \cdot c_i + \sqrt{100(p_i \cdot c_i)\log n}$ edges of $G_i$ are updated in the last $c_i$ updates. Now observe that $c_i \leq \frac{\epsilon}{k} \cdot\frac{\mu_i + 1}{p_i} + 1$ since any time $c_i \geq \frac{\epsilon}{k} \cdot \frac{\mu_i + 1}{p_i}$, we immediately reset $c_i$ to zero in the update algorithm. Therefore, since clearly $c_i \leq n \Delta \leq n^2$, under the high probability event of the previous paragraph,
	\begin{flalign*}
	|F_i| &\leq p_i \cdot (\tfrac{\epsilon}{k} (\tfrac{\mu_i+1}{p_i}) + 1) + \sqrt{100\big(p_i \cdot (\tfrac{\epsilon}{k} (\tfrac{\mu_i+1}{p_i}) + 1)\big)\log n}\\
		&=\tfrac{\epsilon}{k} \mu_i + \tfrac{\epsilon}{k} + p_i + 10\sqrt{(\tfrac{\epsilon}{k} \mu_i + \tfrac{\epsilon}{k} + p_i) \log n}\\
		&\leq \tfrac{\epsilon}{k} \mu_i + 2 + 10\sqrt{(\tfrac{\epsilon}{k} \mu_i + 2) \log n} \tag{Since $p_i \leq 1$ and $\epsilon/k \leq 1$.}\\
		&\leq 2(\tfrac{\epsilon}{k} \mu_i + 2) + 2 \cdot 100 \log n.\\[-12pt] \tag{If $(\frac{\epsilon}{k}\mu_i + 2) \geq 10\sqrt{(\tfrac{\epsilon}{k} \mu_i + 2) \log n}$ then the first term ensures inequality, otherwise the second.}
	\end{flalign*}
	The proof is thus complete.
\end{myproof}

We are now ready to prove \Cref{lem:sum-Di-small}.

\begin{proof}[Proof of \Cref{lem:sum-Di-small}]
	From \Cref{cl:cclcrg128391t-1h2t3} we get that with probability $1-1/n^4$, 
	$$
		\sum_{i=1}^{k+1} |F_i| \leq \sum_{i=1}^{k+1} (2 \tfrac{\epsilon}{k} \mu_i + 200 \log n + 4) \leq 204 (k+1) \log n + 2\frac{\epsilon}{k} \sum_{i=1}^{k+1} \mu_i.
	$$
	Combined with inequality $\sum_{i=1}^{k+1} \mu_i \leq \frac{3(k+1)}{2} (\mu(G) + \sum_{i=1}^{k+1} |F_i|)$ implied by \Cref{cl:clrr-12390123}, we get
	$$
	\sum_{i=1}^{k+1} |F_i| 
	\leq 204(k+1) \log n + \frac{2 \epsilon}{k} \cdot \frac{3(k+1)}{2} \left(\mu(G) + \sum_{i=1}^{k+1} |F_i|\right).
	$$
	Noting that $k \geq 1$, we can simplify and re-arrange the terms, obtaining that
	$$
	\sum_{i=1}^{k+1} |F_i| \leq \frac{204(k+1) \log n + 6 \epsilon \mu(G)}{1-6 \epsilon} \stackrel{(\epsilon \leq 1/12)}{\leq} 408 (k+1) \log n + 12 \epsilon \mu(G) < 13 \epsilon \mu(G),
	$$
	where the last inequality follows from our assumption of \Cref{rem:mularge} that $\mu(G) \geq \mulb$. 
\end{proof}

\subsection{Update Time of the Algorithm of \Cref{sec:the-algorithm}}\label{sec:update-time}

In this section we prove the update-time bound of \Cref{thm:dynamicalg}, except that instead of a worst-case update-time, we here prove an amortized update-time bound. We then show in \Cref{sec:worstcase} how with a small modification this can be turned into a worst-case bound.

The cost of maintaining the maximal matchings of $G_1, \ldots, G_{k+1}$ as stated before is $k \poly(\log n)$ for every update w.h.p. The ``easy updates'' such as removing a deleted edge from any of $U_i$'s or $G_i$'s can also be done in $O(k\log n)$ time as previously discussed. It only remains to analyze the cost of the ``heavy updates''. That is, the calls to $\computeLayers{j}$. 

Let us fix some $j \in [k+1]$ and analyze the (amortized) cost of a call to $\computeLayers{j}$. For the rest of this section, unless otherwise stated explicitly, when we refer to a data structure of the algorithm (such as $H_i$, $\mu_i$, $G$, etc.) we refer to the value of this data structure right after the call to \computeLayers{j}. Note that it takes at least another $\frac{\epsilon}{k} \cdot \frac{\mu_j+1}{p_j}$ updates until we call $\computeLayers{j}$ for this specific value of $j$ again. As such, we can amortize the cost of a call to $\computeLayers{j}$ over at least $\frac{\epsilon}{k} \cdot \frac{\mu_j+1}{p_j}$ updates.\footnote{Note that there is an edge case: If the number of remaining updates is not as large as $\frac{\epsilon}{k} \cdot \frac{\mu_j+1}{p_j}$ then we cannot amortize the cost over the future updates. However, since this happens at most once for each $j$, and since $\computeLayers{j}$ clearly takes at most linear-time in the number of edges of the whole graph, we can amortize this cost over the whole sequence of updates which involves $\Omega(m)$ updates as assumed in \Cref{thm:dynamicalg}.} The resulting amortized update-time summed up for all $j \in [k+1]$ gives the amortized update-time of the algorithm.

 Let us start with an upper bound on the running time of \computeLayers{j}.

\begin{claim}\label{cl:time-computeLayers}
	The time spent in subroutine $\computeLayers{j}$ is $\Ot\big( (\epsilon^{-1} + \beta k) |U_j| + \epsilon^{-1}|H_k| \big)$.
\end{claim}
\begin{myproof}
	We start by analyzing the calls $H_i \gets \addLayer{U_i \cap G_i, H_{i-1}, \mu_i}$ for $i \in \{j, \ldots, k\}$. Note that each graph $U_i \cap G_i$ can be constructed in $\Ot(|U_i|)$ time by iterating over the edges $e$ of $U_i$ and considering the rank $\pi_e$ which determines if $e \in G_i$. Algorithm \addLayer{U_i \cap G_i, H_{i-1}, \mu_i} iterates over the edges in $U_i \cap G_i$ and for each edge that is added to $H_i$, we have to find out if there are any $(H_i, \beta)$-overfull connected to its endpoints. By trivially scanning all the at most $\beta$ neighbors (by \Cref{obs:degree-of-H}) this can be done in $O(|U_i \cap G_i| \beta) = O(|U_i|\beta)$ time. The overall time-complexity of these calls is therefore bounded by
	\begin{flalign*}
		\Ot(|U_j|\beta + \ldots + |U_k|\beta) = \Ot(|U_j|\beta k)\tag{Since $U_j \supseteq \ldots \supseteq U_k$ by \Cref{obs:monotonicity}.}
	\end{flalign*}
	
	Construction of each $U_{i+1}$ from $U_i$ can also be done in $O(|U_i|)$ time by simply iterating over the edges and checking the edge-degree of each edge in $O(1)$ time. Since this is run for $i \in \{j, \ldots, k\}$, the total time is $O(|U_j| + \ldots + |U_k|) = O(|U_j|k)$.
	
	As discussed, each $\mu_i$ takes $O(1)$ time to compute, hence $\mu_j, \ldots, \mu_{k+1}$ take $O(k)$ time to compute.
	
	The final step is to run \Cref{prop:apxmatching} to find a $(1-\epsilon)$-approximation of $(H_k \cup U_{k+1}) \cap G$. We first construct graph $(H_k \cup U_{k+1})$ in $O(|H_k| + |U_{k+1}|)$ time, then iterate over its edges and remove any edge that does not belong to $G$. This can be done easily in $\Ot(|H_k| + |U_{k+1}|)$ time and $\Ot(m)$ space by storing the adjacency lists of $G$ in a BST so that each pair can be checked to be neighbors in $\Ot(1)$ time. Then running \Cref{prop:apxmatching} on the resulting graph requires $O(\epsilon^{-1}(|H_k|+|U_{k+1}|))$ time. Noting that $U_{k+1} \subseteq U_j$ by \Cref{obs:monotonicity}, the overall time of this step is $\Ot(\epsilon^{-1}(|H_k| + |U_j|))$.
		
	Summing up all the mentioned bounds proves the bound of the claim.
\end{myproof}

Recall that the time-complexity of \computeLayers{j} is amortized over $\frac{\epsilon}{k} \cdot \frac{\mu_j+1}{p_j}$ updates. Using the upper bound of \Cref{cl:time-computeLayers}, we  amortized cost of \computeLayers{j} is thus at most:
\begin{equation}\label{eq:nclrcggrcg-12t3}
	O\left( \frac{(\epsilon^{-1}+\beta k)|U_j| + \epsilon^{-1}|H_k|}{\frac{\epsilon}{k} \cdot \frac{\mu_j+1}{p_j}} \right)
= \left(\frac{p_j|U_j|}{\mu_j+1} + \frac{p_j|H_k|}{\mu_j+1}\right) \poly(\epsilon^{-1} \beta k \log n).
\end{equation}

\Cref{cl:first-term} below can be used to bound the first term, and \Cref{cl:second-term} can be used to bound the second term. 

\begin{claim}\label{cl:first-term}
	For every $j \in [k+1]$, with probability $1-1/n^4$, $\frac{p_j |U_j|}{\mu_j + 1} = \Delta^{\frac{1}{k+1}} \cdot \poly(\epsilon^{-1}\beta \log n)$.
\end{claim}

\begin{claim}\label{cl:second-term}
	For every $j \in [k+1]$, with probability $1-1/n^4$, $\frac{p_j |H_k|}{\mu_j + 1} = O(\beta \log n)$.
\end{claim}

The proof of \Cref{cl:first-term} is harder and is carried out in two sections. In \Cref{sec:j>=2} we prove \Cref{cl:first-term} for $j \geq 2$ using a sparsification guarantee on the size of $U_j$ for $j \geq 2$. Then in \Cref{sec:j=1} we prove \Cref{cl:first-term} for $j=1$ using a lower bound on the size of $\mu_1$.

The proof of \Cref{cl:second-term} is simple and we present it in \Cref{thm:dynamicalg}.

Before proving these claims, let us confirm that they do indeed imply the update-time bound of \Cref{thm:dynamicalg} (emphasizing again that we are bounding the amortized update-time here which we turn into a worst-case bound in \Cref{sec:worstcase}).

\begin{proof}[Proof of the update-time of \Cref{thm:dynamicalg}]
	As discussed, all computations outside \computeLayersNoInput{} take $k \poly(\log n)$ worst-case time per update. On the other hand, replacing the bounds of \Cref{cl:first-term} and \Cref{cl:second-term} into \Cref{eq:nclrcggrcg-12t3}, we get that the amortized cost of \computeLayers{j} for every $j \in [k+1]$ is $\Delta^{\frac{1}{k+1}} \cdot \poly(\epsilon^{-1} \beta k \log n)$. Summing all of them up, this only multiplies this bound by a $(k+1)$ factor, which is still $\Delta^{\frac{1}{k+1}} \cdot \poly(\epsilon^{-1} \beta k \log n)$.
	
	We show in \Cref{sec:worstcase} how the algorithm can easily be de-amortized by spreading the cost of a call to \computeLayersNoInput{} over multiple edge-updates, obtaining the claimed update-time bound of \Cref{thm:dynamicalg}.
\end{proof}

\subsubsection{Proof of \Cref{cl:first-term} for $j \geq 2$}\label{sec:j>=2}

For when $j \geq 2$, our main tool is the following sparsification lemma, which bounds the size of subgraph $U_j$. This sparsification property holds because of the special way we construct subgraph $H_{j-1}$ in \Cref{alg:H}. Intuitively, we commit to $H_{j-1}$ in \Cref{alg:H} when many edges have arrived in the order of $\pi$ and none of them are $(H_{j-1}, \beta)$-underfull. Since $\pi$ is a random order of the edge-set of the graph, we can then conclude that w.h.p. there are not so many  $(H_{j-1}, \beta)$-underfull edges left, which is precisely the size of $U_j$ right after we recompute $H_{j-1}$.

Although the details and parameters are very different, we note that the proof of this sparsification property is inspired by a work of Bernstein \cite{BernsteinICALP20} in the random-order streaming model.

\begin{lemma}[\textbf{Sparsification Lemma}]\label{lem:sparsification}
	For any $i \in [k]$, at any given time in the algorithm it holds with probability $1-1/n^{5}$ that
	$$
		|U_{i+1}| = O\left( \frac{\mu_{i} \beta^2 \log n}{p_{i}} \right).
	$$
\end{lemma}
\begin{proof}
	Note that after $H_i$ and $U_{i+1}$ are computed in \Cref{alg:computeLayers}, it takes at most $\frac{\epsilon}{k} \cdot \frac{\mu_i}{p_i}$ other updates to recompute them. During these updates, the size of $U_{i+1}$ increases by at most $\frac{\epsilon}{k} \cdot \frac{\mu_i}{p_i} \leq \frac{\mu_i}{p_i}$ since each edge update either adds at most one edge to $U_{i+1}$. As a result, it suffices to prove that w.h.p. $|U_{i+1}| = O\left( \frac{\mu_{i} \beta^2 \log n}{p_{i}} \right)$ right after a call to \Cref{alg:computeLayers}.
	
	The crux of the proof will be about analyzing the behavior of \addLayerNoInput{} (\Cref{alg:H}) which is called in \Cref{alg:computeLayers} to construct subgraph $H_i$, which in turn, is used to define $U_{i+1}$.
	
	As a thought experiment and only for the sake of the analysis, consider a modified version of \Cref{alg:H}, which we call \modifiedAddLayerNoInput{}, with two changes: $(i)$ instead of $\Gamma$ which will be $U_i \cap G_i$ when \addLayerNoInput{} is called, we iterate over all edges of $U_i$ in a {\em random} order; additionally $(ii)$ $\modifiedAddLayerNoInput{}$ takes a parameter $\tau$ as the input and returns $H_i$ when $\eta \geq \tau$ as opposed to the condition in \Cref{line:termination-condition} of \addLayerNoInput. We will first analyze \modifiedAddLayerNoInput{} and then show how it relates to the actual algorithm \addLayerNoInput{}.
	
	Let us condition on the subgraph $H_i$ after \modifiedAddLayerNoInput{} processes some $t$ edges of $U_i$ and let $y$ be the number of unprocessed edges of $U_i$ that are $(H_i, \beta)$-underfull. The probability that the next edge that arrives is $(H_i, \beta)$-underfull is at least $y/|U_i|$. Moreover, assuming that this edge is not $(H_i, \beta)$-underfull, subgraph $H_i$ does not change and so there are still at least $y$ other $(H_i, \beta)$-underfull unprocessed edges. As a result, the probability that the algorithm processes at least $10(|U_i|/y)\log n$ more edges and none of them are $(H_i, \beta)$-underfull, is  at most
	$$
	(1-y/|U_i|)^{10(|U_i|/y)\log n} \leq e^{-10 \log n} \leq n^{-10}.
	$$
	Equivalently, by a union bound over at most $n^2$ choices of $t$, we get that if \modifiedAddLayerNoInput{} successfully returns subgraph $H_i$, which recall that happens when it processes $\tau$ edges of $U_i$ and does not encounter any $(H_i, \beta)$-underfull edges, then with probability $1-n^{-8}$, the total number of $(H_i, \beta)$-underfull unprocessed edges in $U_i$ is at most $10(|U_i|/\tau) \log n$. By another union bound over at most $n^2$ choices of $\tau$, this holds for {\em every} input $\tau \in [n^2]$ with probability $1-n^{-6}$.
	
	Now let us go back to the actual algorithm \addLayerNoInput{}. We prove the claim even if we condition on subgraph $G_{i-1}$. That is, suppose that all the edges $e$ with $\pi_e \leq p_{i-1}$ are revealed. Note that conditioned on this event, the rank of every edge in $G \setminus G_{i-1}$ is independent and uniformly picked from $(p_{i-1}, 1]$. Now, recall that in \addLayerNoInput{} we process the edges of $U_i \cap G_i$ in the increasing order of $\pi$. Since an edge belongs to $G_i$ iff $\pi_e \leq p_i$ and $p_i < p_{i+1} < \ldots < p_{k+1}$ and since $U_i$ does not include any edge of $G_{i-1}$ by definition, this is equivalent to iterating over the edges of $U_i$ until the next edge $e$ in the sequence has $\pi_e > p_i$, i.e., does not belong to $G_i$ anymore. But if we reach this point in \addLayerNoInput{}, then it means that we have not already returned $H_i$, contradicting \Cref{cor:algreturnsH} that the algorithm always terminates. This implies that the extra condition on the next edge not belonging to $G_i$ is not, in fact, needed. This implies, in turn, that algorithm \modifiedAddLayerNoInput{} is exactly equivalent to \addLayerNoInput{} where the parameter $\tau$ is simply set to $\lfloor |\Gamma|/(4 \mu_i \beta^2 + 1) \rfloor = \lfloor |U_i \cap G_i|/(4 \mu_i \beta^2 + 1) \rfloor$. From the discussion of the previous paragraph, therefore, we can infer that when \addLayerNoInput{} terminates, the number of unprocessed edges in $U_i$ that are $(H_i, \beta)$-underfull (which also includes all edges in $U_{i+1}$) is with probability $1-n^{-6}$ at most
	\begin{equation}\label{eq:cghhhn-123}
		\frac{10 |U_i| \log n}{\lfloor |U_i \cap G_i|/(4 \mu_i \beta^2 + 1) \rfloor}.
	\end{equation}
	
	We finish the proof by considering the two cases $(i)$ $p_i|U_i| \geq \mu_i \beta^2\log n$ and $(ii)$ $p_i|U_i| < \mu_i \beta^2 \log n$ separately. (Case $(ii)$ happens to be trivial.)
	
	Consider case $(i)$ first. Since $U_i \subseteq G \setminus G_{i-1}$ and since each edge of $G \setminus G_{i-1}$ belongs to $G_i$ independently with probability $\frac{p_i - p_{i-1}}{1-p_{i-1}} \geq p_i/2$ (the inequality holds by \Cref{cl:pvalues} part $(v)$) conditioned on $G_{i-1}$, we get $\E[|U_i \cap G_i| \mid G_{i-1}] \geq \Omega(p_i |U_i|)$. Moreover, $p_i|U_i|$ is large enough in case $(i)$ to apply the Chernoff bound and get $|U_i \cap G_i| = \Omega(p_i |U_i|)$ with probability, say, $1-1/n^{-6}$. Combined with \Cref{eq:cghhhn-123} we can bound the size of $|U_{i+1}|$, with probability $\geq 1-n^{-5}$ by
	$$
		|U_{i+1}| \stackeq{(\ref{eq:cghhhn-123})}{\leq} \frac{10 |U_i| \log n}{\lfloor |U_i \cap G_i|/(4 \mu_i \beta^2 + 1) \rfloor} = O\left( \frac{|U_i| \log n}{p_i|U_i|/(\mu_i \beta^2)} \right) = O\left( \frac{\mu_i \beta^2 \log n}{p_i} \right).
	$$
	(We note that we used the assumption of $(i)$ one more time in the first equality above to get that the denominator does not become zero when taking the floor.)
	
	For case $(ii)$, note that $U_{i+1} \subseteq U_i$ from \Cref{obs:monotonicity} and thus we simply get $|U_{i+1}| \leq |U_i| \leq \frac{\mu_i \beta^2 \log n}{p_i}$ where the last inequality uses the assumption of $(ii)$. The proof is thus complete.
\end{proof}

Armed with the sparsification lemma, we can now prove \Cref{cl:first-term} for $j \geq 2$.

\begin{proof}[Proof of \Cref{cl:first-term} for $j \geq 2$.]
	From \Cref{lem:sparsification}, we get that with probability $1-1/n^5$, 
	\begin{equation}\label{eq:llrc-123898999}
	\frac{p_j |U_j|}{\mu_j + 1} = \frac{p_j }{\mu_j + 1} \cdot O\left( \frac{\mu_{j-1} \beta^2 \log n}{p_{j-1}} \right) = O\left( (\Delta^{\frac{1}{k+1}} /\epsilon) \cdot \frac{\mu_{j-1} \beta^2 \log n}{\mu_j + 1}\right),
	\end{equation}
	where the last bound follows from bound $p_j/p_{j-1} = O(\Delta^{\frac{1}{k+1}} / \epsilon)$ for all $j \geq 2$ of \Cref{cl:pvalues} part $(iii)$.
	
	Now note that $\mu_{j-1}$ is set to be the size of a maximal matching of $G_{j-1}$ in \computeLayersNoInput{} and $\mu_j$ is the size of a maximal matching of $G_j$ for the same reason. We would like to say that this means $\mu_{j-1} = O(\mu_j)$ since $G_{j-1} \subseteq G_j$. The only remaining challenge, however, is that we only set $\mu_{j-1}$ to be $\maximalmatchingsize{G_{j-1}}$ in \computeLayersNoInput{} and do not update $\mu_{j-1}$ until \computeLayers{i} is called again for some $i \leq j-1$. As a result, the edge updates since $\mu_{j-1}$ was last computed may cause $\maximalmatchingsize{G_i}$ to get much smaller than $\mu_{j-1}$. However, given that $\mu_{j-1}$ is recomputed after at most $\frac{\epsilon}{k} \cdot \frac{\mu_{j-1} + 1}{p_{j-1}}$ updates and only $p_{j-1}$ fraction of these updates belong to $G_{j-1}$ in expectation, the expected number of removed edges from $G_{i-1}$ is at most $O(\frac{\epsilon}{k} \mu_{j-1}) = O(\epsilon \mu_{j-1})$ and so $\maximalmatchingsize{G_{i-1}}$ is {\em in expectation} still at least $\Omega(\mu_{j-1})$. 
	
	To go from expectation to high probability, let $F_{j-1}$ be the set of edges added/removed from $G_{j-1}$ during the last $c_j$ updates (i.e., since the last time $\mu_{j-1}$ was computed). We showed in \Cref{cl:cclcrg128391t-1h2t3} that $|F_{j-1}| = O(\frac{\epsilon}{k} \mu_{j-1} + \log n)$ w.h.p. As such, under this high probability event, we have  $\mu_j = \maximalmatchingsize{G_j} = \Omega(\maximalmatchingsize{G_{j-1}}) = \Omega(\mu_{j-1} - \frac{\epsilon}{k} \mu_{j-1} - \log n) = \Omega(\mu_{j-1} - \log n)$. This means that, w.h.p., $\mu_{j-1}/(\mu_j + 1) = O(\log n)$. Plugging this to \Cref{eq:llrc-123898999} we get $\frac{p_j |U_j|}{\mu_j + 1} \ll O(\Delta^{\frac{1}{k+1}} \epsilon^{-1} \beta^2 \log^3 n)$.
\end{proof}

\subsubsection{Proof of \Cref{cl:first-term} for $j = 1$}\label{sec:j=1}

For the case where $j = 1$, by definition $U_j = U_1 = G$. Therefore, $U_1$ includes all the edges of the graph. To bound the update-time in this case, we show that $\mu_1$ is sufficiently large.

While it is well-known that any $m$-edge graph of maximum degree $\Delta$ has a matching of size $\Omega(m/\Delta)$, we prove in \Cref{lem:subgraphmatching} below the somewhat surprising fact that essentially the same lower bound of $\Omega(m/\Delta)$ holds for the size of the maximum matching in a random edge-subgraph provided that the edge-sampling probability satisfies a rather mild 	constraint.

\begin{claim}\label{lem:subgraphmatching}
	Let $G=(V, E)$ be an arbitrary $n$-vertex graph, let $\Delta$ be an upper bound on $G$'s maximum degree, and let $G_p=(V, E_p)$ be a random subgraph of $G$ including each edge independently with some probability $p$. If $p \geq \max\{15\frac{\ln n}{\Delta}, 32 \frac{\ln n}{|E|} \}$, then 
	$
		\Pr\left[ \mu(G_p) \geq \frac{|E|}{8\Delta}\right] \geq 1-2/n^4.
	$
\end{claim}
\begin{myproof}
	Since $|E_p|$ is a sum of $|E|$ independent $p$-Bernoulli random variables, we have $\E|E_p| = p|E|$ and by applying the Chernoff bound we get
	\begin{equation}\label{eq:gg1219837}
		\Pr\left[|E_p| < \frac{p|E|}{2}\right] \leq
		\exp\left(- \frac{0.5^2 \E|E_p|}{2} \right) =
		\exp\left(- \frac{p |E|}{8} \right) \stackrel{p \geq \frac{32 \ln n}{|E|}}{\leq} \exp(- 4 \ln n) = n^{-4}.
	\end{equation}
	Next, note that for every vertex $v$, $\deg_{G_p}(v)$ is a sum of $\deg_G(v)$ independent $p$-Bernoulli random variables. This means $\E[\deg_{G_p}(v)] = p \deg_G(v) \leq p \Delta$. Applying Chernoff bound, we therefore get
	\begin{equation}\label{eq:rrll12893}
		\Pr[\deg_{G_p}(v) \geq 2 p \Delta] \leq
		\exp\left(-\frac{p \Delta}{3} \right) \stackrel{p \geq 15 \ln n/\Delta}{\leq}
		\exp\left(-\frac{15 \ln n}{3}\right) = n^{-5}.
	\end{equation}
	By a union bound, the maximum degree $\Delta_p$ of $G_p$ is at most $2p \Delta$ with probability $1-n^{-4}$. 
	
	Now take an arbitrary $2 \Delta_p$ edge coloring of $G_p$ and pick the color class with the largest number of colors. This is a matching of size at least $\frac{|E_p|}{2 \Delta_p}$. Hence, 
	$$
		\mu(G_p) \geq \frac{|E_p|}{2 \Delta_p} \stackrel{(\ref{eq:gg1219837})}{\geq}
		\frac{p|E|/2}{2 \Delta_p} \stackrel{(\ref{eq:rrll12893})}{\geq}
		\frac{p|E|/2}{2 (2 p \Delta)} \geq \frac{|E|}{8\Delta},
	$$
	with probability at least $1-2n^{-4}$.
\end{myproof}

\Cref{lem:subgraphmatching} is all we need to prove \Cref{cl:first-term} for $j=1$.

\begin{proof}[Proof of \Cref{cl:first-term} for $j=1$]
Recall from the statement of \Cref{cl:first-term} that we need to prove $\frac{p_1 |U_1|}{\mu_1 + 1} = O(\Delta^{\frac{1}{k+1}} \beta^2 \log^2 n)$. Since $U_1 = G$ at all times, we would like to apply \Cref{lem:subgraphmatching} and obtain that $\mu_1$ is w.h.p. at least $\Omega(|G|/\Delta) = \Omega(|U_1|/\Delta)$. From this, we would get that $\frac{p_1|U_1|}{\mu_1+1} = O(p_1\Delta)$. Given that $p_1 \leq 15 \Delta^{\frac{1}{k+1} - 1} \log n$, the RHS is $O(\Delta^{\frac{1}{k+1}} \log n)$ which is the desired bound.

To apply \Cref{lem:subgraphmatching} and complete the proof, we only need to show that $p_1 \geq 15 \frac{\ln n}{\Delta}$ and $p_1 \geq 32 \frac{\ln n}{|G|}$. The first inequality is  proved in \Cref{cl:pvalues} part $(ii)$. If the second condition does not hold, i.e., if $p_1 < 32 \frac{\ln n}{|G|}$, then we can prove the claim trivially. To see this, note that $\frac{p_1|U_1|}{\mu_1 + 1} \leq p_1|U_1| = p_1|G|$ and the latter is at most $O(\log n)$ if $p_1 < 32 \frac{\ln n}{|G|}$. Hence, either \Cref{cl:first-term} follows trivially or we can apply \Cref{lem:subgraphmatching} and prove it as discussed above.
\end{proof}

\subsubsection{Proof of \Cref{cl:second-term}}\label{sec:proof-second-term}

In this section we prove \Cref{cl:second-term} that, w.h.p., $\frac{p_j|H_k|}{\mu_j + 1} = O(\beta \log n)$.

We start with a simple observation to bound the size of $H_k$.

\begin{observation}\label{obs:tssssssthn-123}
	$|H_k| \leq 2\mu(G) \beta$.
\end{observation}
\begin{myproof}
	Observe that $G$ has a vertex cover $W$ with size at most $2\mu(G)$ (pick the two endpoints of a maximum matching of $G$). Moreover, since $H_k \subseteq G$ by \Cref{obs:monotonicity} part $(v)$, $W$ is also a vertex cover of $H_k$. Combined with \Cref{obs:degree-of-H} that bounds the maximum degree of $H_k$ by $\beta$, we get that $H_k$ has at most $|W| \beta = 2\mu(G)\beta$ edges.
\end{myproof}

Observe that if $p_j \mu(G) \leq 10\log n$ then we readily have the bound of \Cref{cl:second-term} since, by \Cref{obs:tssssssthn-123}, $\frac{p_j|H_k|}{\mu_j+1} \leq \frac{p_j \mu(G) \beta}{\mu_j+1} \leq p_j \mu(G) \beta = O(\beta \log n)$. So let us assume $p_j \mu(G) \geq 10 \log n$.

Now fix a maximum matching of $G$ and recall that each edge appears in $G_j$ independently with probability $p_i$. As such, $\E[\mu(G_j)] \geq p_j \mu(G)$. With our earlier assumption of $p_j \mu(G) \geq 10 \log n$, we can thus apply the Chernoff bound to get that with probability, say, $1 - 1/n^4$, $\mu(G_j)$ is at least $\Omega(p_j \mu(G))$. Combined with \Cref{obs:tssssssthn-123}, we thus get
$$
\frac{p_j|H_k|}{\mu_j + 1} \leq \frac{p_j \cdot \mu(G) \beta}{\mu_j +1} \leq \frac{p_j \cdot \mu(G) \beta}{\mu_j} = \frac{p_j \cdot \mu(G) \beta}{\Omega(p_j \mu(G))} = O(\beta).
$$
This completes the proof of \Cref{cl:second-term}.

\subsection{Getting a Worst-Case Update-time Bound}\label{sec:worstcase}

The algorithm that we presented in \Cref{sec:the-algorithm} was shown in \Cref{sec:update-time} to have the same update-time as claimed in \Cref{thm:dynamicalg}. However, we analyzed the {\em amortized} update-time in \Cref{sec:update-time} instead of the {\em worst-case} update-time. In this section, we show how with a simple trick of spreading the computation over multiple updates, we can get the same update-time but in the worst-case. We note that this idea is standard and has been used before in \cite{GuptaPeng-FOCS13} and \cite{BernsteinFH-SODA19}. As such, we only give a high level discussion of how it works.

Observe that the only place in the analysis of update-time that we used amortization was in bounding the update-time caused by the calls to \computeLayers{j} for various $j$. Indeed, with the way we defined the algorithm in \Cref{sec:the-algorithm}, this amortization is necessary since when we call \computeLayers{j} the time-complexity is larger than the final update-time and this must be amortized. The trick to get a worst-case bound is to spread this computation over the updates. That is, suppose that \computeLayers{j} takes $T$ operations. Instead of performing all these $T$ operations over one single edge update, we do it over multiple edge updates. 

More formally, recall that in our algorithm, when we call $\computeLayers{j}$, we set $c_j$ to be zero. Then upon each update we increase $c_j$ by one and call $\computeLayers{j}$ again when $c_j$ exceeds $\frac{\epsilon}{k} \cdot \frac{\mu_j + 1}{p_j}$. Now instead, when $c_j$ exceeds half this threshold, we call $\computeLayers{j}$ but spread its computation over the next $0.5 \frac{\epsilon}{k} \cdot \frac{\mu_j + 1}{ p_j}$ updates. Only when this computation is finished, we update the data structures. It is easy to see that the approximation ratio does not hurt since the total ``wait time'' until \computeLayers{j} is called again remains the same. For the update-time, one can adapt essentially the same analysis of the amortized update-time of \Cref{sec:update-time} to show that this modified algorithm now has the same update-time but in the worst-case. 

\hiddencomment{We can provide more details here. Particularly, we are here essentially spreading the computation over "previous" updates, whereas the amortized analysis amortizes them over future updates. There is no difference, since the threshold $\epsilon \mu_j/p_j$ only changes by a constant factor at most from one call to $\computeLayers{j}$ to the next, but this is worth mentioning.}

See \cite{GuptaPeng-FOCS13} and \cite{BernsteinFH-SODA19} for more discussions on this deamortization technique.

\subsection{Bounding Maximum Degree by $O(\sqrt{m})$}\label{sec:degree-root-m}

Up to this point, we have given an algorithm satisfying the approximation guarantee of \Cref{thm:dynamicalg} in $\Delta^{\frac{1}{k+1}} \cdot \poly(\epsilon^{-1} \beta k \log n) = \Ot(\Delta^{\frac{1}{k+1}})$ update-time, whereas observe that we claimed a bound of $\min\{\Delta^{\frac{1}{k+1}}, m^{\frac{1}{2(k+1)}} \} \cdot \poly(\epsilon^{-1} \beta k \log n)$ in \Cref{thm:dynamicalg}. In this section, we show how this is possible.

We prove the following lemma that can be applied as a black-box to the algorithm we have, yielding the guarantee of \Cref{thm:dynamicalg}.

\begin{lemma}
	Consider a fully dynamic graph $G$ and let $\Delta$ and $m$ be fixed upper bounds on the maximum degree and the number of edges of $G$. Suppose that there is an algorithm $\mc{A}$ that maintains an $\alpha$-approximate maximum matching of $G$ in $T(\Delta, n)$ update-time for some function $T(\Delta, n)$. Then there is an algorithm $\mc{A}'$ that maintains a $(1-\epsilon)\alpha$-approximate maximum matching of $G$ in $O(\min\{T(\Delta, n), T(\sqrt{m}/\epsilon, n) \} \log n)$ update-time. If the update-time of $\mc{A}$ is worst-case, then so is that of $\mc{A}'$.
\end{lemma}
\begin{myproof}
	Consider a process where each vertex $v$ in $G$ marks $\Delta' = O(\sqrt{m}/\epsilon)$ of its edges arbitrarily and let $\widetilde{G}$ be the subgraph of $G$ including each edge that is marked by {\em both} of its endpoints. Note that $\widetilde{G}$ clearly has maximum degree $\Delta'$. We show that this subgraph can be maintained in a way that every edge update to $G$ leads to at most three edge updates to $\widetilde{G}$. Additionally, we show that $\widetilde{G}$ will always include a $(1-O(\epsilon))$-approximate maximum matching of $G$.
	
	Let us first show how $\widetilde{G}$ can be maintained in $O(\log n)$ worst-case time, by simply maintaining the marked and not-marked edges of each vertex in a balanced BST. Upon insertion of an edge $e$ we check how many edges each of its endpoints are marked; each one of them that has marked less than $\Delta'$ edges adds $e$ to the set of its marked edges and if both add it we insert $e$ to $\widetilde{G}$. Upon deletion of an edge $e$, if it belongs to $\widetilde{G}$ we remove it, we also remove it from the marked edges of its endpoints. If the number of marked edges of any  endpoint of $e$ goes below the threshold $\Delta'$, it marks a new edge and adds it to $\widetilde{G}$ if it should. Overall, each edge update to $G$ can be handled in $O(\log n)$ time and leads to at most 3 edge updates to $\widetilde{G}$. Additionally, the construction of $\widetilde{G}$ is completely deterministic and so the sequence of updates to $\widetilde{G}$ gets fixed once those of $G$ are fixed.
	
	Now we prove that at any time $\mu(\widetilde{G}) \geq (1-O(\epsilon))\mu(G)$. To show this, we note that the marking algorithm above was first introduced by Solomon \cite{Solomon-ITCS18}. He showed that by setting $\Delta' = O(\alpha/\epsilon)$ where $\alpha$ is the {\em arboricity} of the graph, $\widetilde{G}$ will include a $(1-\epsilon)$-approximate maximum matching of $G$. This is all we need since it is a well-known fact that every $m$-edge graph has arboricity $O(\sqrt{m})$.
	
	To conclude the proof, note that we can simply maintain $\widetilde{G}$ and feed it to algorithm $\mc{A}$. Since the maximum degree of $\widetilde{G}$ is always $O(\sqrt{m}/\epsilon)$ and its matching is nearly as large as $G$, we get the claimed bound.
\end{myproof}

%% file: conclusion.tex
\section{Conclusion \& Open Problems}

We introduced the hierarchical edge-degree constrained subgraph (HEDCS). Using the HEDCS, we gave a unified framework that leads to several new update-time/approximation trade-offs for the fully dynamic matching problem, while also recovering previous trade-offs.

While we provided both a factor revealing LP (\Cref{sec:fk-LP}) and another analytical method (\Cref{sec:analysis-of-fkb}) for analyzing the approximation ratio achieved via HEDCS, it remains an extremely interesting problem to analyze its precise approximation factor. Specifically:
\begin{itemize}
	\item While the approximation ratio of HEDCS can still be tangibly above $1/2$ for say $k=4$, $k=5$, etc., we did not specify any lower bounds on this approximation ratio in \Cref{thm:main} since our factor revealing LP of \Cref{sec:fk-LP} becomes too inefficient to run for $k > 3$. 
	\item In our bounds of \Cref{thm:main} there are gaps between bipartite graphs and general graphs. We conjecture that this gap should not exist and the approximation ratio achieved via \HEDCS{\beta}{k}, for any constant $k$, should converge to the same value for both bipartite and general graphs (by letting parameter $\beta$ to be a large enough constant).
	\item It would be interesting to analyze how fast the approximation ratio of \HEDCS{\beta}{k} converges to $1/2$ by increasing $k$. While our \Cref{lem:large-k} shows that the approximation ratio of any \HEDCS{\beta}{k} is at least $\frac{1}{2} + \frac{1}{2^{2^{O(k)}}}$ (for large enough $\beta$), we do not believe this double-exponential dependence on $k$ is the right answer.
\end{itemize}

Next, we note that while the oblivious adversary assumption is well-received in the literature and holds in many natural applications of dynamic matching, it would be interesting to obtain the new trade-offs that we give also against adaptive adversaries (see \cite{Wajc-STOC20,BhattacharyaK21-ICALP21} for discussions about adaptive adversaries in the context of dynamic matching). One way to achieve this would be to give a deterministic dynamic algorithm for maintaining an HEDCS.

More broadly, the following intriguing questions about dynamic matching still remain open:

\begin{open}
Does there exist a fully dynamic algorithm maintaining a $(\frac{1}{2}+\Omega(1))$-approximate maximum matching in $n^{o(1)}$ update-time? In $\poly(\log n)$ update-time?
\end{open}

\begin{open} 
	Does there exist a fully dynamic algorithm maintaining a $(\frac{2}{3}+\Omega(1))$-approximate maximum matching in $o(n)$ update-time?
\end{open}

\section*{Acknowledgements}

We thank the anonymous SODA'22 reviewers for their thoughtful comments.

%% file: analysis-of-fkb.tex
\section{Approximation Ratio of HEDCS: An Analytical Lower Bound}\label{sec:analysis-of-fkb}

In this section, we give a lower bound on the approximation ratio $\alpha(k, \beta, \beta^-)$ for large $k$. This analysis is particularly useful when the LP approach described in \Cref{sec:fk-LP} (which produces better lower bounds) becomes too inefficient to run in practice. We prove the following:

\begin{lemma}\label{lem:large-k}
	Fix any integers $k \geq 1$ and $\beta > \beta^- \geq 1$ where $\beta^- = (1-\delta)\beta$ for some $0 \leq \delta \leq 0.2$. Define $h(x) := 1/2^{2^{2x}}$. If $h(k) - 4\delta \geq 0$ then
	$$
		\alpha(k, \beta, \beta^-) \geq \frac{1}{2} + \frac{h(k)}{6} - \frac{2}{3} \delta.
	$$
\end{lemma}

Note that by picking $\beta$ large enough and $\beta^-$ close enough to $\beta$, we can make $\delta$ desirably small (even dependent on $k$), satisfying $h(k) - 4 \delta \geq h(k)/2 > 0$. Hence, \Cref{lem:large-k}  together with the approximation guarantee of HEDCS based on function $\alpha(\cdot)$ discussed in \Cref{sec:HEDCS} implies that setting $k = \frac{1}{\epsilon} -1$ and setting $\beta = O_\epsilon(1)$ large enough, results in a \HEDCS{\beta}{k} which includes a strictly better than half approximate matching of ratio $\frac{1}{2} + \Omega_{\epsilon}(1)$ and can be maintained in update-time $\widetilde{O}(\min\{\Delta^\epsilon, m^{\epsilon/2}\})$, as claimed in \Cref{table:results}.

We prove the following auxiliary claim.

\begin{claim}\label{cl:cclhxbhs-th123}
	Let $h(x) := 1/2^{2^{2x}}$. Let $H=(P, Q, E)$ be any bipartite \HEDCS{\beta}{k} with at least $(1-\delta)\beta|P|/2$ edges, where $0 \leq \delta \leq 1 - 3\sum_{i=1}^{k}\sqrt{h(i)}$. Then
	$$
		|Q| \geq \left(1+h(k)-4\delta\right)\frac{|P|}{2}.
	$$
\end{claim}

Let us first see how \Cref{cl:cclhxbhs-th123} proves \Cref{lem:large-k}.

\begin{proof}[Proof of \Cref{lem:large-k}]
	First, it can be confirmed that $1 - 3\sum_{i=1}^{k} \sqrt{h(i)} > 0.2$ for any $k \geq 1$. As a result, the condition $0 \leq \delta \leq 0.2$ of \Cref{lem:large-k} always satisfies the condition $0 \leq \delta \leq 1 - 3\sum_{i=1}^{k}\sqrt{h(i)}$ of \Cref{cl:cclhxbhs-th123}.
	
	Next, note from \Cref{def:f} that since \Cref{cl:cclhxbhs-th123} holds for any bipartite \HEDCS{\beta}{k}, we have
	$$
		f(k, \beta, \beta^-) \geq \frac{1}{2} \cdot (1+h(k)-4\delta).
	$$
	On the other hand, by \Cref{def:alpha} we have $\alpha(k, \beta, \beta^-) = \frac{2f(k, \beta, \beta^-)}{2f(k, \beta, \beta^-)+1}$, which means
	$$
		\alpha(k, \beta, \beta^-) \geq \frac{1+h(k)-4\delta}{2+h(k)-4\delta} \geq \frac{1}{2} +\frac{h(k)-4\delta}{6} = \frac{1}{2} + \frac{h(k)}{6} - \frac{2}{3} \delta.
	$$
	The second inequality above comes from the fact that $\frac{1+x}{2+x} \geq \frac{1}{2} + \frac{x}{6}$ for any $0 \leq x \leq 1$.
\end{proof}

So it remains to prove \Cref{cl:cclhxbhs-th123}.

\begin{proof}
	The function $h$ in the statement is specifically defined in a way that for any integer $k \geq 2$, it satisfies the following condition which will be useful later in the proof 
	\begin{equation}\label{eq:hcllrch9123-123}
		h(k-1) - 12\sqrt{h(k)} \geq h(k).
	\end{equation}
	Let us, for brevity, define $d_P := m/|P| = (1-\delta)\beta/2$ to be the average degree of the $P$ side of $H$. Let us also assume that $H_1, \ldots, H_k$ is a hierarchical decomposition for \HEDCS{\beta}{k} which must exist by definition of HEDCS.
	
	We now prove the claim by induction on $k$.
	
	\paragraph{Base case $k=1$:} Since for $k=1$ a \HEDCS{\beta}{1} is by definition equivalent to a $\beta$-EDCS, from the known bounds for EDCS we have $|Q| \geq \frac{d_P}{\beta - d_P}|P|$ (see e.g. \cite[Lemma 2.2]{BehnezhadEDCS21} or \cite{AssadiBernsteinSOSA19}). Thus:
	$$
		|Q| \geq \frac{d_P}{\beta - d_P}|P| = \frac{(1-\delta)\beta/2}{\beta-(1-\delta)\beta/2} |P| = \frac{1-\delta}{1+\delta} |P| \geq (1-2\delta) |P| = (2-4\delta)\frac{|P|}{2} > (1+\frac{1}{2^{2^2}} - 4\delta)|P|.
	$$
	
	\paragraph{Induction step:} Now let us assume that the claim holds for $k-1$ and prove it for $k$. Let us for brevity define $\xi := \sqrt{h(k)}$ and assume for the sake of contradiction that 
	\begin{equation}\label{eq:hclrchh-9102398}
	|Q| < (1+h(k)-4\delta)\frac{|P|}{2} = (1+\xi^2 - 4\delta)\frac{|P|}{2}.
	\end{equation}
	Define $Q_T := \{ w \in Q \mid \deg_H(w) \leq (1-\xi)\beta\}$. Since the maximum degree of any HEDCS is bounded by $\beta$ as proved in \Cref{obs:HEDCS-max-degree}, we have
	$$
		(|Q| - |Q_T|) \cdot \beta + |Q_T| \cdot (1-\xi)\beta \geq m \geq (1-\delta) \beta |P|/2.
	$$
	The LHS can be simplified to $\beta |Q| - \xi \beta |Q_T|$. Then by canceling the $\beta$ terms on both sides, we get
	$$
		|Q| - \xi |Q_T| \geq (1-\delta)|P|/2.
	$$
	This, in turn, implies that
	\begin{equation}\label{eq:clr921987-9}
		|Q_T| \leq \frac{|Q| - (1-\delta)\frac{|P|}{2}}{\xi} 
		\stackeq{(\ref{eq:hclrchh-9102398})}{\leq} \frac{(1+\xi^2-4\delta)\frac{|P|}{2} - (1-\delta)\frac{|P|}{2}}{\xi} = \frac{\xi^2-3\delta}{\xi}\cdot \frac{|P|}{2} \stackeq{$(\delta \geq 0)$}{\leq} \xi \cdot \frac{|P|}{2}.
	\end{equation}
	
	Now let $P_k$ be the vertices in $P$ that have at least one edge in $H_k \setminus H_{k-1}$ with the other endpoint in $Q \setminus Q_T$. Observe that by definition of $Q_T$, for any $v \in Q \setminus Q_T$ we have $\deg_H(v) > (1 - \xi)\beta$. As such, for property $(i)$ of HEDCS to hold for $H$, every vertex in $P_k$ must have degree at most $\xi \beta$. (As otherwise, its edge in $H_k \setminus H_{k-1}$ that goes to $Q \setminus Q_T$ violates property $(i)$ of HEDCS.)
	
	We remove all the {\em vertices} in $Q_T$ (along with their edges) from graph $H$. We also remove all the {\em edges} of $P_k$ from $H$, leaving their vertices as singletons in the graph. Observe that this removes all the edges of $H_k \setminus H_{k-1}$ from the graph since any such edge must be either connected to $P_k$ or to $Q_T$, and so we end up with a \HEDCS{\beta}{k-1} and can apply the induction hypothesis. Before applying it, though, let us upper bound the total number of removed edges. Since as discussed every vertex in $H$ has degree at most $\beta$, and every vertex in $P_k$ has degree at most $\xi \beta$, 	
	\begin{flalign*}
		\text{\# of removed edges} \leq |Q_T| \cdot \beta + |P_k| \cdot \xi \beta \stackeq{(\ref{eq:clr921987-9})}{\leq} \xi \frac{|P|}{2} \cdot \beta + |P_k| \cdot \xi \beta \leq \frac{3}{2} \xi |P| \beta.
	\end{flalign*}
	Hence, the resulting \HEDCS{\beta}{k-1} still has $|P|$ vertices in one part, and the number of its edges is
	$$
		(1-\delta)\beta |P|/2 - \frac{3}{2} \xi |P| \beta = (1 - \delta  - 3 \xi)\beta |P| /2.
	$$
	We have $\delta \leq 1 - 3 \sum_{i=1}^{k} \sqrt{h(i)} = 1 - 3\sum_{i=1}^{k-1} \sqrt{h(i)} - 3 \xi$ from the statement of the claim. This means that, first, $1 - \delta - 3\xi \geq 3\sum_{i=1}^{k-1} \sqrt{h(i)} \geq 0$, and second, $\delta + 3\xi \leq 1 - 3 \sum_{i=1}^{k-1}\sqrt{h(i)}$. Hence, the number of edges of the resulting \HEDCS{\beta}{k-1} satisfies the constraint of the lemma  (i.e., the induction hypothesis) with $\delta' = \delta + 3 \xi$ and so we get
	\begin{flalign*}
		|Q| \geq |Q \setminus Q_T| &\geq \left( 1 + h(k-1) - 4(\delta + 3\xi) \right) \frac{|P|}{2}\\
		&= (1 + h(k-1) - 12 \xi - 4\delta) \frac{|P|}{2}\\
		&= (1 + h(k-1) - 12 \sqrt{h(k)} - 4 \delta) \frac{|P|}{2}\\
		&\geq (1 + h(k) - 4\delta) \frac{|P|}{2}. \tag{By (\ref{eq:hcllrch9123-123}).}
	\end{flalign*}
	This is exactly the needed inequality for $k$ and so the proof is complete.
\end{proof}

%% file: appendix.tex
\section{Proof of \Cref{prop:apx-HEDCS}}\label{sec:prop-apx-HEDCS-proof}

\begin{restate}{\Cref{prop:apx-HEDCS}}
	\propapxstatement
\end{restate}

As discussed, the proof of is obtained by adapting the technique of Assadi and Bernstein \cite{AssadiBernsteinSOSA19} for the analysis of EDCS. Particularly, the analysis for the bipartite case is completely due to \cite{AssadiBernsteinSOSA19}, and the analysis of the general case is also based on their approach, with an additional idea.

It would be convenient to consider a slightly more definition of HEDCS. 

\begin{definition}\label{def:generalized-HEDCS}
	Let $\beta > \beta^- \geq 1$ and $k \geq 1$ be integers. We say  graph $H$ is a \HEDCS{(\beta, \beta^-)}{k} of $G$ iff there is a hierarchical decomposition $\emptyset = H_0 \subseteq H_1 \subseteq H_2 \subseteq \ldots \subseteq H_k = H$ satisfying:
\begin{enumerate}[label=$(\roman*)$]
	\item For every $i \in [k]$ and any edge $e \in H_i \setminus H_{i-1}$, $\deg_{H_i}(e) \leq \beta$.
	\item For any edge $e \in G \setminus H$, $\deg_{H}(e) \geq \beta^-$.
\end{enumerate}
\end{definition}

Note that only property $(ii)$ has changed from the original \Cref{def:HEDCS} of HEDCS.  In particular, a \HEDCS{\beta}{k} is simply a \HEDCS{(\beta, \beta-1)}{k}.

We start with the following lemma that directly follows from the approach of \cite{AssadiBernsteinSOSA19}:

\begin{lemma}\label{lem:HEDCS-bipartite}
Let $H$ and $U$ be subgraphs of a bipartite graph $G$ and suppose $H$ is a \HEDCS{(\beta, \beta^-)}{k} of $G \setminus U$. Then $\mu(H \cup U) \geq \alpha(k, \beta, \beta^-) \mu(G)$.
\end{lemma}
\begin{myproof}[Proof sketch]
	The proof is essentially identical to the proof of \cite[Lemma~3.1]{AssadiBernsteinSOSA19} and we follow the same terminology. Let $L$ and $R$ be the vertex parts of the base graph $G$. Consider an extended Hall's witness set $A$ for graph $H \cup U$, suppose w.l.o.g. that $A \subseteq L$ and let $B = N_{H \cup U}(A) \subseteq R$. Also define $\bar{A} := L \setminus A$ and $\bar{B} := R \setminus \bar{B}$.
	
	There must be a matching $M$ in $G$ going from $A$ to $\bar{B}$ such that $|M| = \mu(G) - \mu(H \cup U)$. Additionally, it must hold that $\mu(H \cup U) \geq |\bar{A}| + |B|$ (see \cite{AssadiBernsteinSOSA19} Eq (1)). 
	
	Let $H'$ be the subset of edges of $H$ with one endpoint matched by $M$. Let $P$ be the subset of vertices of $H'$ that touch $M$ and let $Q$ the rest of the vertices of $H'$. Note that $H'$ must be bipartite with $P$ and $Q$ being a valid partitioning.  Now it can be confirmed that $H'$ has at least $|P|\beta^-/2$ edges because every edge in $M$ is missing from $H \cup U$ and so its edge-degree by the second property of HEDCS must be at least $\beta^-$ in $H$. From the definition of $f(k, \beta, \beta^-)$ applied to graph $H'$, we get that the $Q$ side of $H'$ has at least $f(k, \beta, \beta^-)|P| = 2f(k, \beta, \beta^-)|M|$ vertices. But since $Q \subseteq \bar{A} \cup B$, this implies that
	$$
		\mu(H \cup U) \geq |\bar{A}| + |B| \geq 2f(k, \beta, \beta^-)|M| \geq 2f(k, \beta, \beta^-)(\mu(G) - \mu(H \cup U)).
	$$
	Moving the terms, we get 
	$$
		\mu(H \cup U) \geq \frac{2f(k, \beta, \beta^-)}{2f(k, \beta, \beta^-) + 1} \mu(G) = \alpha(k, \beta, \beta^-) \mu(G).\qedhere
	$$
\end{myproof}

Note that \Cref{lem:HEDCS-bipartite} immediately implies the first part of \Cref{prop:apx-HEDCS} for bipartite graphs. This can also be extended to the general case using the L\'ovasz Local Lemma (LLL) as in \cite{AssadiBernsteinSOSA19}, leading to the following result (which note is slightly different from \Cref{prop:apx-HEDCS}):

\begin{lemma}\label{lem:general-via-LLL}
	Let $H$ and $U$ be subgraphs of a general (i.e., not necessarily bipartite) graph $G$ and suppose that $H$ is a \HEDCS{\beta}{k} of $G \setminus U$. Then for some $\gamma = O(\sqrt{\beta \log(k\beta)})$, it holds that $\mu(H \cup U) \geq \alpha(k, \frac{1}{2}\beta + \gamma, \frac{1}{2}\beta - \gamma ) \mu ( G )$.
\end{lemma}

\begin{myproof}
	Fix an arbitrary maximum matching $M^\star$ of $G$ and then construct a random bipartite subgraph $\widetilde{G}=(L, R, \widetilde{E})$ of $G$ by putting one endpoint of each $e \in M^\star$ in $L$ and the other in $R$ chosen randomly, and allocating the rest of the vertices not matched by $M^\star$ independently and uniformly to $L$ and $R$. This way, the maximum matching of $\widetilde{G}$ remains exactly equal to $G$. Now let $\widetilde{H} := H \cap \widetilde{G}$, let $(H_1, \ldots, H_k)$ be the hierarchical decomposition of $H$, and let $\widetilde{H}_i = H_i \cap \widetilde{G}$.
	
	As in \cite{AssadiBernsteinSOSA19}, the key observation is that the expected degree of every vertex $v$ in $\widetilde{H}_i$ is essentially $\deg_{H_i}(v)/2$ up to an additive error of one. Letting $\lambda = 6\sqrt{\log(k\beta)/\beta}$, the Chernoff bound gives
	$$
		\Pr\Big[|\deg_{\widetilde{H}_i}(v) - \deg_{H_i}/2| \geq \lambda \beta + 1\Big] \leq 2\exp\left( -\frac{\lambda^2 \beta^2}{3\beta} \right) = 2\exp(-12\log(k\beta)) < (k\beta)^{-10}.
	$$
	
	Now let $\mc{E}_{v, i}$ be the event that $|\deg_{\widetilde{H}_i}(v) - \deg_{H_i}/2| \geq \lambda \beta + 1$. Note that $\mc{E}_{v, i}$ depends only on the part that vertices in $N_H(v)$ are assigned to. Noting that $|N_H(v)| \leq \beta - 1$ for all $v$ by \Cref{obs:HEDCS-max-degree}, we get that $\mc{E}_{v, i}$ depends on at most $\beta^2 k$ other events $\mc{E}_{u, j}$. As such, the events $\{\overline{\mc{E}_{v, i}}\}_{v \in V, i \in [k]}$ satisfy the LLL condition, and so there is a partitioning where $\cap_{v \in V, i \in [k]} \overline{\mc{E}_{v, i}}$ happens. We consider this partitioning from this point on in the analysis.
	
	Let $\gamma := 2 \lambda \beta + 3 \leq 13 \sqrt{\beta \log(k\beta)}$. We show that $\widetilde{H}$ is a \HEDCS{(\frac{1}{2}\beta + \gamma, \frac{1}{2}\beta - \gamma)}{k} of $\widetilde{G}$. Indeed, we show that $\widetilde{H}_1, \ldots, \widetilde{H}_k$ is a hierarchical decomposition, satisfying the HEDCS properties for $\widetilde{H}$. We just have to prove the two properties of HEDCS holds.

	For property $(i)$ of HEDCS, for any $i \in [k]$ and any edge $e = (u, v) \in \widetilde{H}_i$ we have $$\deg_{\widetilde{H}_i}(e) \leq \frac{1}{2}\deg_{H_i}(e) + 2\lambda \beta +2 \leq \frac{1}{2}\beta+2\lambda \beta +2 \leq \frac{1}{2}\beta + \gamma$$ where the first inequality holds because of the events $\overline{\mc{E}_{v, i}}$ and $\overline{\mc{E}_{u, i}}$, and the second inequality holds because $H$ is a \HEDCS{\beta}{k}. 
	
	For property $(ii)$ of HEDCS, take an edge $e = (u, v) \in \widetilde{G} \setminus \widetilde{H}$. We have $$\deg_{\widetilde{H}}(e) \geq \frac{1}{2}\deg_{H}(e) - 2\lambda \beta - 2 \geq \frac{1}{2}(\beta -1) - 2\lambda\beta -2 > \frac{1}{2}\beta - \lambda,$$ where again the first inequality holds because of the events $\overline{\mc{E}_{v, i}}$ and $\overline{\mc{E}_{u, i}}$, and the second inequality holds because $H$ is a \HEDCS{\beta}{k} and $e \in G \setminus H$ by construction.
	
	Plugging \Cref{lem:HEDCS-bipartite}, we thus get $\mu(\widetilde{H} \cup U) \geq \alpha(k, \frac{1}{2}\beta + \gamma, \frac{1}{2}\beta-\gamma) \mu(\widetilde{G})$. The claim then follows since $\mu(H \cup U) \geq \mu(\widetilde{H} \cup U)$ as $\widetilde{H} \subseteq H$, and since $\mu(G) = \mu(\widetilde{G})$ as discussed.
\end{myproof}

To go from the guarantee of \Cref{lem:general-via-LLL} to the guarantee of \Cref{prop:apx-HEDCS}, we prove the following claim using the {\em dependent rounding} scheme of Gandhi {\em et al.}~\cite{GandhiKPS06}.

\begin{claim}
	Let $k \geq 1$ and $\beta \geq 1$ be integers, let $\beta$ be such that $\beta \geq \generalbeta{\beta' k}$ for some sufficiently large constant $c \geq 1$ and integer $\beta' \geq 2k - 1$, and let $\gamma$ be as in \Cref{lem:general-via-LLL}. Then it holds that $\alpha(k,\frac{1}{2}\beta+\gamma, \frac{1}{2}\beta-\gamma) \geq (1 - o(1)) \cdot \alpha(k,\beta'+2k-1,\beta')$.
\end{claim}

\begin{myproof}
It suffices to prove
\begin{equation}\label{eq:rc019283-20}
f(k, \frac{1}{2}\beta + \gamma, \frac{1}{2} \beta - \gamma) \geq (1-2\epsilon) \cdot f(k, \beta' + 2k-1, \beta').
\end{equation}
To see why (\ref{eq:rc019283-20}) suffices, observe that $\frac{2f}{2f+1} \geq \frac{2f'}{2f'+1}$ for any $f \geq f' > 0$. Thus from (\ref{eq:rc019283-20}) and from \Cref{def:alpha} we get
\begin{flalign*}
	\alpha(k, \frac{1}{2}\beta + \gamma, \frac{1}{2} \beta - \gamma) &= \frac{2f(k, \frac{1}{2}\beta + \gamma, \frac{1}{2} \beta - \gamma)}{2f(k, \frac{1}{2}\beta + \gamma, \frac{1}{2} \beta - \gamma) + 1} \geq \frac{2(1-2\epsilon)f(k, \beta' + 2k-1, \beta')}{2(1-2\epsilon)f(k, \beta' + 2k-1, \beta')+1}\\
	&\geq (1-2\epsilon)\frac{2f(k, \beta' + 2k-1, \beta')}{2f(k, \beta' + 2k-1, \beta')+1} = (1-2\epsilon)\alpha(k, \beta'+2k-1, \beta'),
\end{flalign*}
which is the needed guarantee of the claim. So from now on, we focus on proving (\ref{eq:rc019283-20}).

Take any bipartite graph $G = (P \cup Q, E)$ that is a \HEDCS{(\frac{1}{2}\beta + \gamma)}{k} with at least $|E| \geq (\frac{1}{2} \beta - \gamma)|P|/2$ edges. We construct from $G$ another bipartite graph $G' = (P \cup Q, E')$ that is a \HEDCS{(\beta' + 2k-1)}{k} with at least $|P| \beta'/2$ edges. \Cref{eq:rc019283-20} then immediately follows from \Cref{def:f}.

To go from $G$ to $G'$, we first construct a fractional solution $G_1$ obtained by assigning a weight of $\beta'/(\frac{1}{2}\beta + \gamma)$ to each edge in $G$. This ensures that $G_1$ is a {\em fractional} \HEDCS{\beta'}{k} (where the edge-degree constraints are generalized to fractional edge-degrees) with total weight at least
$$
\frac{\beta'}{\tfrac{1}{2}\beta+\gamma} \cdot |E| \geq \frac{\beta'}{\tfrac{1}{2}\beta+\gamma} \cdot (\tfrac{1}{2} \beta - \gamma)|P|/2 \geq (1-o(1)) \beta' |P|/2.
$$
Note that the average degree of the vertices in $P$ is now lower bounded by $(1-o(1))\beta'/2$ instead of $\beta'/2$ but it is straightforward to show that for any $\beta_1 > \beta_2 > 1$, we have $\alpha(k, \beta_1, (1 - o(1))\beta_2) \geq (1 - o(1)) \cdot \alpha(k, \beta_1, \beta_2)$. So we will now simply focus on showing that we can round the fractional solution above, into an integral one containing at least as many edges such that the degrees of end-points of any edge increases by at most $2k - 1$ as a result of the rounding.

Our rounding algorithm proceeds in iterations such that in the iteration $i$ of rounding, we round all edges whose level is $i$. In doing this rounding, we will ensure that the degree of an edge at a level higher than $i$ can increase by at most $2$ due to the rounding of level $i$ edges. As such this implies that after all $k$ levels have been rounded, the resulting integral solution $G'$ satisfies the property that all edge degrees are bounded by $\beta' + 2k$. However, with just a slightly more careful analysis, we will be able to improve this to $\beta' + (2k - 1)$, giving us the desired result.

We now describe the rounding scheme. We start by considering edges at level 1. We can view this as a fractional solution $x^{(1)}$ (obtained by restricting the fractional solution to only level $1$ edges) that induces a fractional degree $x^{(1)}(u)$ at each vertex $u$. Now using the dependent rounding scheme in Theorem 2.3 of~\cite{GandhiKPS06}, we know that there exists an integral solution that has (i) at least as many edges as the weight of $x$, and (ii) ensures that the degree of each vertex $u$ is either $\lfloor x^{(1)}(u) \rfloor$ or $\lceil x^{(1)}(u) \rceil$. Thus we can obtain an integral solution where the degree of each vertex goes up by only $1$ while preserving the total fractional mass of edges at level $1$. This means that any edge at levels $1$ through $k$ sees an increase of at most $2$ in the total degree of its end-points as a result of this rounding. We can now repeat this process on edges at levels $2$ through $k$, ultimately obtaining an integral solution such that degree of each edge in the final rounded solution can be bounded by $\beta' + 2k$.

To improve the bound to $\beta' + (2k-1)$, we observe the following. Fix any edge $(u,v)$ at some level $i$. If each of $x^{(1)}(u), x^{(2)}(u), ..., x^{(i)}(u)$, $x^{(1)}(v), x^{(2)}(v), ..., x^{(i)}(v)$ are integral, then the dependent rounding scheme leaves these degrees unaltered, and hence the degree of the edge $(u,v)$ in the rounded solution continues to be bounded by $\beta'$. On the other hand, if at least one of $x^{(1)}(u), x^{(2)}(u), ..., x^{(i)}(u)$, $x^{(1)}(v), x^{(2)}(v), ..., x^{(i)}(v)$ is fractional, then it must be that before rounding, the degree of the edge $(u,v)$ satisfies

$$ \sum_{j=1}^{i} \left( \lfloor x^{(j)}(u) \rfloor  + \lfloor x^{(j)}(v) \rfloor \right) \le \beta' - 1.$$ 
It then follows that after rounding, the degree of the edge $(u,v)$ is at most
$$ \sum_{j=1}^{i} \left( \lceil x^{(j)}(u) \rceil  + \lceil x^{(j)}(v) \rceil \right) \le \beta' - 1 + (2i) \le \beta' + (2k-1),$$ 
since $i \le k$. As discussed, this completes the proof.
\end{myproof}

\section{Proof of \Cref{cl:potential}}\label{sec:proofs}

\begin{restate}{Claim~\ref{cl:potential}}
	\Cref{alg:H} reaches \Cref{line:addToH-algH} at most  $4\mu_i\beta^2$ times.
\end{restate}

\begin{myproof}
	\Cref{alg:H} reaches \Cref{line:addToH-algH} every time that it encounters an $(H_i, \beta)$-underfull edge. Hence, it suffices to show that the total number of $(H_i, \beta)$-underfull edges encountered cannot is $\leq 2\mu_i\beta^2$. Our proof combines a potential function defined previously in \cite{BernsteinSteinSODA16,AssadiBBMSSODA19,AssadiBernsteinSOSA19,BernsteinICALP20} (see e.g., \cite[Proposition~2.4]{AssadiBernsteinSOSA19} or \cite[Lemma~4.2]{BernsteinICALP20}) with a simple additional idea.
	
	Define the following potential functions
	$$
		\Phi_1 := (2\beta - 1) |H_i|, \qquad \Phi_2 := \sum_{e \in H_i} \deg_{H_i}(e), \qquad \Phi := \Phi_1 - \Phi_2.
	$$
	
	We first show that $\Phi$ is non-negative at the start of \Cref{alg:H} after we set $H_i \gets H_{i-1}$. Recall from \Cref{obs:degree-of-H} that the maximum degree in $H_{i-1}$ is at most $\beta$. This means that $\Phi_2 \leq \beta |H_{i-1}|$, implying
	$$
	\Phi \geq (2\beta - 1)|H_{i-1}| - \beta|H_{i-1}| = (\beta - 1)|H_{i-1}| \geq 0.
	$$ 
	
	The key observation is that every time that we insert an $(H_i, \beta)$-underfull edge to $H_i$ or remove an $(H_i, \beta)$-overfull edge from $H_i$, the value of $\Phi$ increases by at least 1. The proof of this part is completely due to \cite{BernsteinSteinSODA16,AssadiBBMSSODA19,AssadiBernsteinSOSA19,BernsteinICALP20} and proceeds as described next.
	
	Suppose that we remove an $(H_i, \beta)$-overfull edge $e$ from $H_i$. This reduces $\Phi_1$ by exactly $2 \beta - 1$. Let us now analyze the change to $\Phi_2$. On the one hand, removing $e$ from $H_i$ decreases $\Phi_2$ by at least $\beta + 1$ because $e$ no longer participates in the sum and its edge-degree before removing it was at least $\beta + 1$ for being $(H_i, \beta)$-overfull. On the other hand, $e$ must have had at least $\deg_{H_i}(e) - 2 \geq \beta - 1$ incident edges in $H_i$ before being deleted. Deleting $e$ reduces the edge-degree of each of these edges by one, and so $\Phi_2$ overall decreases by $\beta + 1 + \beta - 1 = 2\beta$. Since $\Phi_1$ decreases by exactly $2\beta - 1$ and $\Phi_2$ decreases by at least $2\beta$, $\Phi$ increases by at least 1.
	
	Now consider inserting an $(H_i, \beta)$-underfull edge $e$ into $H_i$. This increases $\Phi_1$ by exactly $2\beta - 1$. Now on the one hand, inserting $e$ increases $\Phi_2$ by at most $\beta$ because $e$ will now participate in the sum and its edge-degree before inserting it was at most $\beta - 2$ for being $(H_i, \beta)$-underfull. On the other hand, $e$ has at most $\beta - 2$ incident edges in $H_i$ (or else it would not have been $(H_i, \beta)$-underfull before being inserted) and adding $e$ increases their edge-degrees by one. Therefore in total $\Phi_2$ increases by at most $\beta + (\beta - 2) = 2\beta - 2$. Since $\Phi_1$ increases by exactly $2\beta - 1$ and $\Phi_2$ increases by at most $2\beta - 2$, the value of $\Phi$ increases by at least 1 in this case too.
	
	In the next step of the proof, we show that $\Phi \leq 2 \mu_i \beta^2$ at all times. From \Cref{obs:monotonicity} recall that $H_i \subseteq G_i$. Since $\mu_i$ is the size of a maximal matching of $G_i$, there exists a vertex cover $W$ of $G_i$ (and thus $H_i$) with size $|W| = 2\mu_i$. Since each vertex in $H_i$ has maximum degree $\beta$ by \Cref{obs:degree-of-H} and each of these edges are connected to $W$, we get $|H_i| \leq 2\mu_i \beta$. This in particular implies $\Phi_1 \leq (2\beta - 1)2\mu_i \beta \leq 4 \mu_i \beta^2$. By non-negativity of $\Phi_2$, we get $\Phi \leq 4\mu_i \beta^2$.
	
	To finish the proof, note on the one hand that in every iteration that we encounter an $(H_i, \beta)$-underfull edge, we add it to $H$ and thus the value of $\Phi$ increases by at least 1 by our discussion above. On the other hand, the value of $\Phi$ is upper bounded by $4 \mu_i \beta^2$ as discussed. Hence, the number of $(H_i, \beta)$-underfull edges encountered by Algorithm~\ref{alg:H} is $\leq 4 \mu_i \beta^2$, which as discussed at the beginning, completes the proof.
\end{myproof}